\documentclass[11pt]{article}
\usepackage{graphicx}
\usepackage{amsfonts}
\usepackage{amsmath}
\usepackage{amssymb}
\usepackage{url}
\usepackage{fancyhdr}
\usepackage{indentfirst}
\usepackage{enumerate} 
\usepackage{amsthm}
\usepackage{color} 
\usepackage{natbib}
\usepackage{dsfont} 
\usepackage[colorlinks=true,citecolor=blue]{hyperref}

\usepackage[misc]{ifsym}

\def\d{\,\mathrm{d}}
\def\laweq{\buildrel \d \over =}

\newcommand{\VaR}{\mathrm{VaR}}

\newcommand{\RVaR}{\mathrm{RVaR}}

\newcommand{\ES}{\mathrm{ES}}
\newcommand{\E}{\mathbb{E}}
\newcommand{\R}{\mathbb{R}}

\newcommand{\p}{\mathbb{P}}

\newcommand{\id}{\mathds{1}}

\newcommand{\X}{\mathcal X}

\newcommand{\A}{\mathsf A}

\newcommand{\M}{\mathcal{M}} 

\renewcommand{\P}{\mathcal{P}}

\newcommand{\esssup}{\mathrm{ess\mbox{-}sup}}
\newcommand{\essinf}{\mathrm{ess\mbox{-}inf}}
\renewcommand{\ge}{\geqslant}
\renewcommand{\le}{\leqslant}
\renewcommand{\geq}{\geqslant}
\renewcommand{\leq}{\leqslant}
\renewcommand{\epsilon}{\varepsilon}
\theoremstyle{plain}
\newtheorem{theorem}{Theorem}
\newtheorem{corollary}[theorem]{Corollary}
\newtheorem{lemma}[theorem]{Lemma}
\newtheorem{proposition}[theorem]{Proposition}
\theoremstyle{definition}
\newtheorem{definition}[theorem]{Definition}
\newtheorem{example}[theorem]{Example}
\newtheorem{assumption}[theorem]{Assumption}

\newtheorem{remark}[theorem]{Remark}
\theoremstyle{definition}

\numberwithin{equation}{section} \numberwithin{theorem}{section}
\renewcommand{\cite}{\citet}

\DeclareMathOperator*{\interior}{int}
\DeclareMathOperator*{\argmin}{arg\,min}
\usepackage[onehalfspacing]{setspace}

\begin{document}
\title{{Elicitability and identifiability of tail risk measures}}

\author{Tobias Fissler\thanks{RiskLab, Department of Mathematics, ETH Zurich, Switzerland. \Letter~{\scriptsize 
\texttt{tobias.fissler@math.ethz.ch}}}\and Fangda Liu\thanks{Department of Statistics and Actuarial Science, University of Waterloo, Canada.  \Letter~{\scriptsize \texttt{fangda.liu@uwaterloo.ca}}}\and Ruodu Wang\thanks{Department of Statistics and Actuarial Science, University of Waterloo, Canada. \Letter~{\scriptsize \texttt{wang@uwaterloo.ca}}} \and Linxiao Wei\thanks{School of Mathematics and Statistics, Wuhan University of Technology, China. \Letter~{\scriptsize   \texttt{lxwei@whut.edu.cn}}}}

\date{\today} 
\maketitle
\begin{abstract}
%\noindent
Tail risk measures are fully determined by the distribution of the underlying loss beyond its quantile at a certain level, with {Value-at-Risk, Expected Shortfall and Range Value-at-Risk} being prime examples.
They are induced by law-based risk measures, called their generators, evaluated on the tail distribution. 
This paper establishes joint identifiability and elicitability results of tail risk measures together with the corresponding quantile, provided that their generators are identifiable and elicitable, respectively.
As an example, we establish the joint identifiability and elicitability of the tail expectile together with the quantile.
The corresponding consistent scores constitute a novel class of weighted scores, nesting the known class of scores of Fissler and Ziegel for the Expected Shortfall together with the quantile.
For statistical purposes, our results pave the way to easier model fitting for tail risk measures via regression and the generalized method of moments, but also model comparison and model validation in terms of established backtesting procedures.

\vspace{0.5em}
\noindent
\begin{bfseries}Key-words\end{bfseries}:  
Tail risk measures,
elicitability, 
identifiability,
weighted scores,
backtesting

\vspace{0.5em}
\noindent
\begin{bfseries}MSC Classification\end{bfseries}: 
91G70;
62F07;
62G32;
62P05

\end{abstract}

\section{Introduction}\label{sec:1}

%\noindent
Over the past three decades, 
 tail risk measures, such as Value-at-Risk (VaR) and Expected Shortfall (ES), have been playing a prominent role 
 as the standard risk metrics  in global  banking and insurance regulatory frameworks, such as Basel III/IV \citep{BASEL19}  and Solvency II \citep{E11}. 
 Within the  general framework of tail risk measures  developed by 
\cite{LW21} and \cite{LMWW22},  there are many   tail risk measures in the literature, such as the tail standard deviation \citep{FL06},
the tail entropic risk measure \citep{T09},
the Glue VaR \citep{BGS14},
the Gini Shortfall \citep{FWZ17},   and 
the Range VaR \citep{CDS10,ELW18}, in addition to  VaR and ES. 
 
In view of the practical importance of tail risk measures, 
this paper studies their identifiability and elicitability. 
A risk measure is identifiable if it is the unique zero (root) of an expected identification function. 
Similarly, it is elicitable if it is the unique minimizer of an expected score; see Definitions \ref{def:identifiability} and \ref{def:elicitability}.
As such, identifiability and elicitability facilitate Z- and M-estimation \citep{HuberRonchetti2009}. More generally, model fitting can be pursued via the generalized method of moments \citep{NeweyMcFadden1994} or regression \citep{DFZ24}, also using recent tools from machine learning; see \cite{FMW23} and \cite{ContETAL2022} for examples.
For model comparison and model validation, commonly subsumed under the umbrella term \emph{backtesting}, identifiability and elicitability are crucial \citep{NoldeZiegel2017, FZG16, WWZ25}.

Elicitability, together with the closely related concept of identifiability, has received increasing attention in the statistics and risk management literature.  
In particular, characterization results on risk measures are obtained via elicitability or convex level sets; 
  see \cite{W06}, \cite{BB15} and \cite{DBBZ16} for convex risk measures, \cite{Z16} for coherent risk measures, \cite{KP16} and \cite{WZ15}   for distortion risk measures, \cite{WW20} for non-monotone risk measures, and \cite{LW21} for tail risk measures. 
Elicitability can also be used indirectly to identify ES among coherent risk measures, as done by \cite{EMWW21}.

A tail risk measure is determined by the tail of the loss distribution and a generating risk measure, called the generator. Our paper is motivated by the intriguing observation  that the mean is elicitable, and so is the vector of VaR and ES at the same level (\citealp{FZ16}), where ES is the tail risk measure generated by the mean. This naturally leads to the question whether the elicitability of a tail risk measure is connected to the elicitability of its generator.  
This problem, as well as its variants, will be systemically studied in this paper.

Our main results, Theorems \ref{thm:identifiability} and \ref{thm:elicitability}, establish the joint identifiability and elicitability of a tail risk measure together with the corresponding  quantile, subject to identifiability and elicitability of its generator, respectively.
For the elicitability result, a novel monotonicity condition on the consistent scoring function for the generator is required.
Interestingly, the consistent scores of the latter theorem constitute a generalization of the weighted scores discussed by \cite{HolzmannKlar2017}, but with a forecast dependent weight function.
This class nests also the well known FZ-scores for ES and the quantile, characterized by \cite{FZ16}. 
 {Vice versa,} Propositions \ref{prop:identifiability} and \ref{prop:elicitability} show that the identifiability and elicitability of the generator is necessary for the joint identifiability and elicitability results, respectively.

The paper is structured as follows. 
Section \ref{sec:2} collects preliminaries on tail risk measures. 
Section \ref{sec:cond-eli} defines the key concepts of elicitability, identifiability, and convex level sets, as well as some examples. 
Section \ref{sec:4} contains results on identifiability, and 
Section \ref{sec:5} contains results on elicitability. 
Some remarks on our results when changing from the right tail to the left tail, or to a body part of the distribution, are presented in Section \ref{sec:sign conventions}.
Section \ref{sec:6} concludes the paper. 
Some background on risk measures,
technical assumptions, and omitted proofs are collected in three appendices.

\section{Preliminaries}\label{sec:2}

 \subsection{Notation}

Let $\M^0$ be the set of all Borel probability distributions on $\R$ and let $\M^q\subseteq \M^0$, $q\in[1,\infty)$, be the subset of all distributions with a finite $q$th moment.
Moreover, $\M^\infty$ denotes the set of all compactly supported distributions. 
For the ease of presentation, we identify distributions with the corresponding cumulative distribution functions; that is, $ F(x) $ is understood as $ F((-\infty , x] ) $ for $ x \in \R $. For $ F \in \M^0 $, we define the left-continuous generalized inverse (left-quantile) as
$$
   F^{-1} (t) = \inf \{ x \in \R: F(x) \geq t \} , \quad t \in (0,1] , 
$$
and in addition set $F^{-1}(0)=\inf\{x\in \R: F(x)> 0\}$. 
For a function $F$ on a subset of $\R$, its left and right limits at $x $ are given by
$F(x-)= \lim_{y\uparrow x}F(y)$ and $F(x+)= \lim_{y\downarrow x}F(y)$, assuming that the limits are well-defined. 
For $x\in \R$, denote by $\delta_x$ the point-mass probability measure at $x$. 
For simplicity,  $\int$ means integration on $\R$ unless otherwise specified.  
Throughout this article, if we do not specify $ \M$, then the statements hold for any set $ \M \subseteq  \M^0 $ such that the functional at consideration is well-defined and finite on $\M$.

\subsection{Risk measures}

{A risk measure is a mapping 
$\rho \colon \M\to [-\infty,\infty]$, where  the set $\M\subseteq \M^0$ is the domain of $\rho$. 
The typical interpretation of a risk measure is that its argument, $F\in \M$, represents the distribution of a random loss from some financial position, and that $ \rho(F) $ represents the amount of capital required to hold to make the financial position acceptable. 
In the literature, risk measures are often defined as mappings from a set of random variables, instead of $\M$, to the extended real line.
When such risk measures are {law-based}, they one-to-one correspond to risk measures under our definition.
We explain in Appendix \ref{app:Risk measures} details on risk measures and their desirable properties.}

The two most popular classes of  {risk measure}s used in banking and insurance practice are the Value-at-Risk (VaR)  and the Expected Shortfall (ES).  
For a confidence level $p $,   the  {right-quantile}  ($\VaR^+_p$) and the  {left-quantile} ($\VaR_p^-$), are defined for $F\in\M^0$ as
\begin{align*}
\VaR^+_p(F)   &= \inf\{x\in \R: F(x)> p\}=F^{-1}(p+) ,~~~p\in [0,1);\\
\VaR^-_p(F)   &=\inf\{x\in \R: F(x)\ge p\}=F^{-1}(p)  ,~~~p\in (0,1].
\end{align*}
In addition, let $ \esssup(F)= \VaR_1^-(F) $ and $ \essinf(F) = \VaR_0^+(F)$ for $F \in \M^0$. 
In risk management practice, one typically does not distinguish between $\VaR_p^+$ and $\VaR^-_p$
as they are identical for distributions with a  continuous $ F^{-1} $ at $p$. 
However, to allow for the full generality of exposition, for $p\in(0,1)$, we shall mostly work with the interval-valued $p$-quantile
$$
Q_p(F) = \{x\in \R: F(x-)\le p \le F(x)\} = \left[\VaR_p^-(F),\VaR_p^+(F)\right].
$$
For a confidence level $ p \in (0,1) $, $\ES_p$  is defined by 
\begin{equation}
\label{eq:ES}
\ES_p(F)=\frac{1}{1-p}\int_p^1 \VaR^+_r (F)\d r \in \R\cup\{\infty\},\quad F\in \mathcal M^0 ,
\end{equation}
$ \ES_1 (F) = \VaR_1^- (F) = \esssup (F)\in \R\cup\{\infty\} $ and $\ES_0(F)$ is the mean on $\M^1$.
We remark that for $ p \in (0,1) $ and $F\in\M^0$, $\ES_p(F)$ is finite if and only if $\int_0^{\infty} x\d F(x)<\infty$.
Since $\VaR_r^-(F) \neq \VaR_r^+(F)$ only for at most countably many $r\in [0,1]$, we can also replace $\VaR^+_r (F)$ by $\VaR^-_r (F)$ in \eqref{eq:ES}.

{
As a compromise between the robustness of VaR and and the coherence of ES, \cite{CDS10} introduced the so called Range Value-at-Risk (RVaR). 
For $0< p<q< 1$, it is defined as 
\begin{equation}
    \RVaR_{p,q}(F)=\frac{1}{q-p}\int_p^q \VaR^+_r (F)\d r \in \R,\quad F\in \mathcal M^0.
\end{equation}
Clearly, for the boundary cases of $q=1$ or $p=q$, by taking limits, $\RVaR_{p,q}$ coincides with $\ES_p$ or $\VaR_p^+$, respectively.
}

\subsection{Tail distribution and tail {risk measure}s}

We follow the same definitions of tail distributions as in \cite{RU02}. 
 For a distribution $ F \in \M^0 $ and $p\in (0,1)$, let  $F_p\in \M^0$ be the
\emph{tail distribution of $F$} beyond its $p$-quantile, that is,
$$F_p (x) = \frac{(F(x) - p )_+}{1-p} , \quad x \in \R .$$ 
Invoking our sign convention that a distribution $ F \in \M^0 $ is a {loss} distribution, it is clear that the right tail is the region of interest from a risk measurement and management perspective \citep{QRM15}.
Clearly, with another sign convention and from a mathematical perspective, we could similarly consider the left tail of the distribution, and all the results presented in the article hold \emph{mutatis mutandis} (this will be made more clear in Section \ref{sec:sign conventions}).
We impose the tacit condition on our generic classes $\M\subseteq \M^0$ that 
  $ F_p\in \M$ for each $ F \in \M  $ and $p\in(0,1)$. This assumption holds for  common choices of $\M $, such as $\M=\M^q$, $q\in [1,\infty]$.\footnote{A notable exception, however, is $\M\setminus \M^\infty$.}
Note that $ F_p = G_p $ if and only if $ F(x) = G(x) $ for all $ x \geq \VaR_p^+ (F) $.

The following definition corresponds to Definition 3.1 of \cite{LW21}.
\begin{definition}[Tail risk measures] \label{def:tail risk fcn}
 For $p\in (0,1)$, a {risk measure} $\rho\colon\M\to\R$ is a \emph{$p$-tail {risk measure}}  if $\rho(F)=\rho(G)$ for all $F,G\in \M $ satisfying $F_p = G_p$. 
{A {risk measure} $\rho \colon \M\to\R$ is  % \emph{tail relevant}, and we call it a
a \emph{tail {risk measure}} 
if it is a $p$-tail risk measure for some $p\in (0,1)$.}
\end{definition}

{It is immediate from Definition \ref{def:tail risk fcn} that for $p,q\in (0,1)$ where $p<q$, $\VaR^+_p$, $\ES_p$ and $\RVaR_{p,q}$ are $r$-tail risk measures for $r\in (0,p]$, but not for $r\in (p,1)$;  $\VaR^-_p$ is an $r$-tail risk measure only for $r\in (0,p)$.}
The expectation is not a tail risk measure. 
An important property of tail risk measures is that they can be generated by other risk measures. 
Precisely, for any $p\in (0,1)$ and $p$-tail risk measure $\rho$ on $ \M $, there exists a risk measure $\rho^*$ on $\M $, called the \emph{$p$-generator} of $ \rho $, satisfying 
\begin{align} \label{def:generator}
\rho(F)=\rho^*(F_p), \quad \text{ for all } F \in \M, 
\end{align}  
and such a $\rho^*$ is unique on the set $\M^{*}$ of distributions in $\M $ whose support is bounded from below (Proposition 1 of \citealp{LW21}). 
Conversely, to obtain a $p$-tail {risk measure}, $p\in(0,1)$, it suffices to specify a generic {risk measure} $\rho^*$ and to define $ \rho $ via \eqref{def:generator}.  
We will define $(\rho,\rho^*)$ as a pair, and study their joint properties.

\begin{definition}
For $p\in (0,1)$,
a pair of {risk measure}s $(\rho,\rho^*)$  is called a \emph{$p$-tail pair} on $\M$ if   
$\rho(F )=\rho^*(F_p)$ for all $F \in \M$, and 
$(\rho,\rho^*)$ is called a \emph{tail pair}  on $\M$ if it is a $p$-tail pair  on $\M$ for some $p\in (0,1)$. 
\end{definition}

A slightly different approach of measuring tail risk is through the loss part of the random variables instead of the part above a quantile level; see \cite{J02} and \cite{CDH13}. We focus on the formulation in Definition \ref{def:tail risk fcn} as it includes the most popular risk measures VaR and ES and many other examples of interest.

\section{Elicitability, identifiability and convex level sets} \label{sec:cond-eli}

In the literature, the general discussion about elicitability and its importance due to its connection to comparative backtests and forecasts can be found in \cite{LPS08},  \cite{G11}, \cite{FZ16}, and the references therein. 
In particular, \cite{NoldeZiegel2017} elaborated in detail on its relevance in backtesting, and they provide calibration tests exploiting the notion of identifiability.
More recently, \cite{LW21} provided a characterization of VaR  among tail risk measures based on  elicitability. 
Moreover, both identifiability and elicitability can be used in estimation, and particularly elicitability also in regression \citep{HuberRonchetti2009, Koenker2005, DFZ24}.

We want to define the notions of identifiability, elicitability and two convex level sets (CxLS) properties in a unified manner, that is applicable for real-valued risk measures, interval-valued quantiles, and vectors consisting of (possibly multiple) real-valued risk measures and an interval-valued quantile.
To this end, we study set-valued functionals $T\colon \M\to\P(\R^k)$, where $\P(\R^k)$ denotes the power set of $\R^k$.
Moreover, we use the tacit convention to identify $x\in\R^k$ with the singleton $\{x\}\in\P(\R^k)$, such that the definitions apply also to $\R^k$-valued functionals.

\begin{definition}\label{def:identifiability}
For a functional $T\colon \M\to\P(\R^k)$, 
a function $V\colon\R^k\times\R\to\R^k$ is an 
\emph{$\M$-identification function} for $T$ if for all $F\in\M$
\begin{align} 
\label{eq:identifiable}
T(F) \subseteq \left\{x\in\R^k: \int_{\R} V(x,y) \d F(y) =0 \right\},
\end{align}
assuming that the integral is well-defined.
If \eqref{eq:identifiable} holds with an equality, then $V$ is 
a \emph{strict $\M$-identification function} for $T$.
If $T$ has a strict $\M$-identification function, we call it \emph{$\M$-identifiable}.
\end{definition}

\begin{definition}\label{def:elicitability}
\begin{enumerate}[(i)]
\item 
For a functional $T\colon \M\to\P(\R^k)$, 
a function $S\colon \R^k\times \R\to\R$ is an \emph{$\M$-consistent score} if for all $F\in\M$
\begin{align} 
\label{eq:elicitable}
T(F) \subseteq  \argmin_{x\in \R^k} \int_{\R} S(x,y) \d F(y) ,
\end{align}
assuming that the integral is well-defined.
If \eqref{eq:elicitable} holds with an equality, then $S$ is a \emph{strictly $\M$-consistent score} for $T$.
If $T$ has a strictly $\M$-consistent score, we call it 
\emph{$\M$-elicitable}.
\item 
A functional $T\colon \M\to\P(\R^{k})$,
 is $ \mathcal{M}$-\emph{conditionally elicitable} with a functional $T'\colon \M\to\P(\R^{k'})$ if $T'$ is $\M$-elicitable and $T$ is $\M(r,T')$-elicitable for any $r\in\R^{k'}$, where
$$
\M(r,T') = \{F\in\M :  r\in T'(F)\}.
$$
\end{enumerate}
\end{definition}

 The notion of conditional elicitability simplifies the original definition given in \cite{EmmerKratzTasche2015}; see also
 \cite{FisslerHoga2023}.
 
 \begin{example}
   The mean functional is $\M^1$-identifiable with the strict $\M^1$-identification function $V(x,y) = x-y$, $x,y\in\R$. 
   Under mild conditions, any other strict $\M^1$-identification function is of the form $h(x)(x-y)$, where $h\neq0$ \citep{DFZ23}. Moreover, under similar conditions, any strictly $\M^1$-consistent scoring function for the mean is given by 
 $$
 S(x,y) = -\phi(x) + \phi'(x)(x-y) + a(y), \quad x,y\in\R,
 $$
 where $\phi\colon\R\to\R$ is strictly convex with subgradient $\phi'$, and $a\colon\R\to\R$ is an $\M$-integrable function \citep{G11}. This class nests the ubiquitous squared loss $S(x,y) = (x-y)^2$, which is strictly $\M^2$-consistent for the mean.
 \end{example}
 
\begin{example}
For $p\in(0,1)$, the $p$-quantile $Q_p$ is identifiable on 
\begin{equation*}
\M_{(p)} = \{F\in\M^0: F(\VaR_p^-(F))=p\}
\end{equation*}
with the strict $\M_{(p)}$-identification function
$
V(x,y) =  \id_{\{ y \le x \}} - p, \ x,y\in\R.
$
Essentially, all other strict $\M_{(p)}$-identification functions are given by $h(x)(\id_{\{ y \le x \}} - p)$, where $h\neq0$; see \cite{DFZ23} for details. Here, ``essentially'' means that for any other strict $\M_{(p)}$-identification function $\tilde{V}(x,y)$ it holds that $\int \tilde{V}(x,y) \d F(y)= h(x)(F(x) - p)$ for all $x\in\R$ and $F\in \M_{(p)}$.
Interestingly, $Q_p$ is elicitable on the entire $\M^0$. Subject to integrability and richness conditions, any strictly $\M^0$-consistent scoring function is given by 
$$
S(x,y) = \id_{\{y>x\}}g(y) + \big(\id_{\{ y \le x \}} - p\big)g(x)  + a(y), \quad x,y\in\R,
$$
where $g$ is strictly increasing \citep{Gneiting2011b}. 
This recovers the well-known pinball loss $S(x,y) = \big(\id_{\{ y \le x \}} - p\big)(x-y)$,
which is strictly $\M^1$-consistent for $Q_p$.
\end{example}

 \begin{example}
 It is known that $\ES_p$ generally fails to be identifiable and elicitable on sufficiently rich classes $\M$ \citep{W06, G11}.
 However, the pair $(Q_p, \ES_p)$ turns out to be $\M^1\cap \M_{(p)}$-identifiable using the strict identification function
 \begin{equation}
 \label{eq:V VaR ES}
 V(v,x,y) = \begin{pmatrix}
 \id_{\{ y \le v \}} - p \\
 x - \tfrac{1}{1-p}\big[ \id_{\{y>v\}}y +(\id_{\{ y \le v \}} - p)v\big]
 \end{pmatrix} , \quad v,x,y\in\R,
 \end{equation}
 see \cite{DFZ23} for the full class of identification functions.
 \cite{AS14} established the $\M$-elicitability of $(Q_p,\ES_p)$ under certain restrictive conditions on $\M$;
 and \cite{FZ16, FZ21} showed that the pair is generally $\M^1$-elicitable. Strictly $\M^1$-consistent scoring functions are of the form 
 \begin{align}\label{eq:S ES VaR}
          \begin{aligned} S(v,x,y) &=  \,  \id_{\{ y > v \}} g (y) + \left ( \id_{\{ y \le v \}} - p \right )g ( v)   \\ 
           & \qquad  +\phi'(x)\Big(x - \tfrac{1}{1-p}\big[ \id_{\{y>v\}}y +(\id_{\{ y \le v \}} - p)v\big]\Big) - \phi(x) + a(y) , \quad v,x,y\in\R,
           \end{aligned}
\end{align}
where $\phi\colon\R\to\R$ is strictly convex with subgradient $\phi'$, such that for all $x$ the function $v\mapsto g(v) - \frac{1}{1-p}\phi'(x)v$ is strictly increasing, and where $a\colon\R\to\R$ is a function such that the expectation $\int S(v,x,y)\d F(y)$ is finite for all $F\in\M^1$.
 \end{example}

\begin{example}
 \label{ex:RVaR}
      {Similarly to $\ES_p$, \cite{FZ21} showed that $\RVaR_{p,q}$, $0<p<q<1$, is not identifiable or elicitable on sufficiently rich classes $\M$.
     In the same vein as before, they established the $\M_{(p)}\cap\M_{(q)}$-identifiability of the triplet $(Q_p,Q_q,\RVaR_{p,q})$ via the identification function 
      \begin{equation}
 \label{eq:V RVaR}
 V(v_1,v_2,x,y) = \begin{pmatrix}
 \id_{\{ y \le v_1 \}} - p \\
 \id_{\{ y \le v_2 \}} - q\\
 x - \tfrac{1}{q-p}\Big[ \id_{\{v_1<y\le v_2\}}y 
 +(\id_{\{ y \le v_1 \}} - p)v_1 
 - (\id_{\{ y \le v_2 \}} - q)v_2\Big]
 \end{pmatrix}
 \end{equation}
 for $v_1,v_2,x,y\in\R$.
 Moreover, strictly $\M^0$-consistent scoring functions for this triplet are of the form 
 \begin{align}\label{eq:S RVaR}
          \begin{aligned} 
         & S(v_1,v_2,x,y) \\ & =  \,  \id_{\{ y > v_1 \}} g_1 (y) + \left ( \id_{\{ y \le v_1 \}} - p \right )g_1 ( v_1) +
          \id_{\{ y > v_2 \}} g_2 (y)   + \left ( \id_{\{ y \le v_2 \}} - q \right )g_2 ( v_2)  \\&  \quad ~  +\phi'(x)\Big(x - \tfrac{1}{q-p}\Big[ \id_{\{v_1<y\le v_2\}}y 
 +(\id_{\{ y \le v_1 \}} - p)v_1 
- (\id_{\{ y \le v_2 \}} - q)v_2\Big]\Big) - \phi(x) + a(y)   
           \end{aligned}
\end{align}
for 
$v_1,v_2,x,y\in\R$,
where $\phi\colon\R\to\R$ is strictly convex with subgradient $\phi'$, such that for all $x$ the functions 
$v_1\mapsto g_1(v_1) - \frac{1}{q-p}\phi'(x)v_1$ 
and 
$v_2\mapsto g_2(v_2) + \frac{1}{q-p}\phi'(x)v_2$
are strictly increasing, and where $a\colon\R\to\R$ is a function such that the expectation $\int S(v_1,v_2,x,y)\d F(y)$ is finite for all $F\in\M^0$.}
 \end{example}
 
 We close this section by providing two versions of the CxLS property, defined in \cite{FFHR2021}.
 
 \begin{definition}
 \begin{enumerate}[(i)]
\item 
 A functional $T\colon \M\to\P(\R^k)$ satisfies  the CxLS property on $\M$ if for all $F_0,F_1\in\M$ and for all $\lambda\in(0,1)$ such that $(1-\lambda)F_0+\lambda F_1 \in\M$ it holds that 
 $$
 T(F_0)\cap T(F_1) \subseteq T((1-\lambda)F_0+\lambda F_1).
 $$
 \item
  A functional $T\colon \M\to\P(\R^k)$ satisfies the CxLS* property on $\M$ if for all $F_0,F_1\in\M$ and for all $\lambda\in(0,1)$ such that $(1-\lambda)F_0+\lambda F_1 \in\M$ it holds that 
 $$
 T(F_0)\cap T(F_1) \neq \emptyset \implies  T(F_0)\cap T(F_1) =  T((1-\lambda)F_0+\lambda F_1).
 $$
 \end{enumerate}
 \end{definition}
 
 Note that \cite{FFHR2021} coined these two properties the \emph{selective} CxLS and CxLS* properties. 
 Since we do not make use of the counterpart -- the \emph{exhaustive} notion -- we omit the qualifier ``selective'' in this paper.
 
 Obviously, the CxLS* property implies the CxLS property. 
 They both generalize the CxLS property for singleton-valued functionals, for which they coincide. 
 Importantly for us, elicitability implies the CxLS* property \citep[Proposition 3.4]{FFHR2021}, and similarly, identifiability implies the CxLS property (Proposition B.4 of the preprint version of \citealp{FisslerHoga2023}). 
 This necessity has already partially been established by \cite{Osband1985} and \cite{G11}.
 Under additional regularity conditions and for real-valued functionals, \cite{SteinwartETAL2014} also showed the sufficiency of the CxLS property for elicitability and identifiability.

 \section{Identifiability results}
 \label{sec:4}
 %\noindent 
We first define the following sets:
\begin{align*}
          \M^{\rm c}_p & = \left \{ F \in \mathcal{M}^0: F^{-1} \text{ is continuous at } p \right \} ,  \quad \text{ for } p \in (0,1) ;\\
          \mathcal M_{\ge r} & =\left\{F\in \mathcal{M}^0: F^{-1}(0)\ge r\right\} , \quad \text{ for } r\in \R ; \\
          \M_{(p)} &= \left\{F\in\M^0: F(F^{-1}(p))=p\right\},  \quad \text{ for } p \in [0,1).
\end{align*}

The first result in this section discusses how the identifiability of the tail risk measure implies that of the corresponding generator.

\begin{proposition}
\label{prop:identifiability}
Let $(\rho,\rho^*)$ be a $p$-tail pair for some $p\in(0,1)$. 
The following statements hold.
\begin{enumerate}[(i)]
\item 
For any $r\in\R$ it holds that if $\rho$  is $\M\cup\{(1-p)G + p\delta_r : G\in\M\}$-identifiable with a strict identification function $V$, then $\rho^*$ is $\M_{\ge r}\cap \M$-identifiable with the strict identification function 
\begin{equation}
\label{eq:V*}
V_r^*(x,y) = (1-p)V(x,y) + pV(x,r), \quad x,y\in\R.
\end{equation}
\item
For any $r\in\R$ it holds that if $(Q_p, \rho)$ is $\M_{(p)}\cap \big(\M\cup\{(1-p)G + p\delta_r : G\in\M\}\big)$-identifiable with a strict identification function
$$\big(\id_{\{y\le v\}} - p,
V(v,x,y)\big),~~~~v,x,y\in\R,$$
for some function $V\colon\R^3\to\R$, 
then $\rho^*$ is $\M_{\ge r}\cap \M_{(0)}\cap \M$-identifiable with the strict identification function
\begin{equation}
\label{eq:V*2}
V_r^*(x,y) = (1-p) V(r,x,y) + pV(r,x,r), \quad x,y\in\R.
\end{equation}
 \end{enumerate}
\end{proposition} 
\begin{proof}
\begin{enumerate}[(i)]
\item
Let $V$
be a strict $\M\cup\{(1-p)G + p\delta_r : G\in\M\}$-identification function for $\rho$ and $V_r^*$ as given in \eqref{eq:V*} for some $r\in\R$.
Choose some $G\in\M_{\ge r}\cap \M$ and define 
$$F(y) = \id_{\{y\ge r\}}\big((1-p)G(y)+p\big) = (1-p) G(y) + p\id_{\{y\ge r\}}  .$$
Due to the assumed identifiability of $\rho$ and since $F_p = G$, we get
\begin{align*}
\rho^*(G) = \rho^*(F_p) = \rho(F) 
&= \left \{x\in\R : \int V(x,y) \d F(y) = 0\right \}\\
&= \left \{x\in\R : \int V(x,y) \d \left ( \id_{\{y\ge r\}}(1-p)G(y) \right ) + pV(x,r) = 0\right \}\\
&= \left \{x\in\R : \int \big ( (1-p) V(x,y) + pV(x,r) \big ) \d G(y) = 0\right \}\\
&= \left \{x\in\R : \int V^*_r(x,y) \d G(y) = 0\right \},
\end{align*}
where $V^*_r(x,y)$ is defined in \eqref{eq:V*}.  
This shows the claim.

\item
Let $\big(\id_{\{y\le v\}} - p,
V(v,x,y)\big)$, $v,x,y\in\R$, be a strict $\M_{(p)}\cap \big(\M\cup\{(1-p)G + p\delta_r : G\in\M\}\big)$-identification function for $(Q_p,\rho)$ and $V_r^*$ given in \eqref{eq:V*2} for some $r\in\R$.
For any $F\in \M_{(p)}\cap \big(\M\cup\{(1-p)G + p\delta_r : G\in\M\}\big)$ and $v\in Q_p(F)$ it holds that 
$$
\int V(v,x,y)\d F(y)=0 \quad \text{ if and only if } \quad x=\rho(F).
$$
Let $G\in\M_{\ge r}\cap \M_{(0)}\cap \M$. Then, as in part (i), define $F(y) = \id_{\{y\ge r\}}\big((1-p)G(y) + p\big)$, resulting in $F\in \M_{(p)}\cap \big(\M\cup\{(1-p)G + p\delta_r : G\in\M\}\big)$ and $r = \VaR_p^-(F)\in Q_p(F)$. 
Again, $F_p = G$, which implies 
\begin{align*}
\rho^*(G) = \rho^*(F_p) = \rho(F) 
&= \left \{x\in\R : \int V(r,x,y) \d F(y) = 0\right \}\\
&= \left\{x\in\R : \int V(r,x,y) \d \left (\id_{\{x\ge r\}}(1-p)G(x)\right ) + pV(r,x,r) = 0\right \}\\
&= \left \{x\in\R : \int \big((1-p) V(r,x,y) + pV(r,x,r)\big ) \d G(x) = 0\right \},
\end{align*}
which shows the claim. \qedhere
\end{enumerate}
\end{proof}

\begin{example}\label{exmp:identifiability}
We illustrate the part (i) of Proposition \ref{prop:identifiability}.
Let $\M$ be the class of continuous distribution functions which are strictly increasing on their support.
For any $\alpha\in(0,1)$, $p\in(0,\alpha]$ and $r\in\R$,
$\VaR_\alpha^-$ is $\M\cup\{(1-p)G + p\delta_r : G\in\M\}$-identifiable with the strict identification function $V(x,y) = \id_{\{y\le x\}} - \alpha$. Moreover, $\VaR_\alpha^-$ is a $p$-tail risk measure with generator $\VaR_{\alpha^*}^-$, where $\alpha^* = (\alpha-p)/(1-p)$. 
Then,  \eqref{eq:V*} yields $V_r^*(x,y) = (1-p)\big(\id_{\{y\le x\}} - \alpha\big) + p \big(\id_{\{y\le r\}} - \alpha\big)$. 
Indeed, for any $F\in\M_{\ge r} \cap \M$ it holds that 
$$
\int V_r^*(x,y) \d F(y) = (1-p)(F(x) - \alpha) + p(1-\alpha) = 0 \quad \text{ if and only if } \quad F(x) = \frac{\alpha - p}{1-p} = \alpha^*,
$$
showing that $V_r^*$ is a strict $\M_{\ge r}\cap\M$-identification function for $\VaR_{\alpha^*}^-$.
\end{example}

The next result addresses the inverse direction of Proposition \ref{prop:identifiability}, that is, how the identifiability of $\rho^*$ implies that of $\rho$, together with the quantile interval. 
\begin{theorem}
\label{thm:identifiability}
Let $(\rho,\rho^*)$ be a $p$-tail pair for some $p\in(0,1)$.
If $\rho^*\colon\M\to\R$ is $\M$-identifiable with the strict identification function $V^*\colon \R\times\R\to\R$, then $(Q_p,\rho)$ is $\M_{(p)}\cap\M$-identifiable with the strict $\M_{(p)}\cap\M$-identification function 
\begin{equation}
\label{eq:identification function}
V(v,x,y) = \begin{pmatrix}
 \id_{\{ y \le v \}} - p \\
\id_{\{y>v\}}V^*(x,y) +(\id_{\{ y \le v \}} - p)V^*(x,v)
 \end{pmatrix}, \quad v,x,y\in\R.
\end{equation}
\end{theorem}

\begin{proof}
Let $F\in \M_{(p)}\cap\M$. 
Clearly, $\int \left( \id_{\{ y \le v \}} - p\right ) \d F(y) = 0$ 
if and only if $v\in Q_p(F)$. 
Take $v\in Q_p(F)$. Since $F_p(y) = \frac{(F(y)-p)^+}{1-p}$, we have
\begin{align*}
 \int \left(\id_{\{y>v\}}V^*(x,y) +(\id_{\{ y \le v \}} - p)V^*(x,v)\right ) \d F(y)
&= \int \id_{\{y>v\}}V^*(x,y)  \d F(y)\\
&= (1-p)\int V^*(x,y)  \d F_p(y).
%&\Longleftrightarrow \quad x = \rho^*(F_p) = \rho(F).
\end{align*}
Hence, $ \int V^*(x,y)  \d F_p(y) = 0 $ if and only if $ x = \rho^*(F_p) = \rho(F)$. 
This shows the claim. 
\end{proof}

\begin{remark}\label{rem:identifiability}
Theorem \ref{thm:identifiability} holds also by using 
\begin{equation}
\label{eq:identification function2}
V(v,x,y) = \begin{pmatrix}
 \id_{\{ y \le v \}} - p \\
\id_{\{y>v\}}V^*(x,y) 
 \end{pmatrix}, \quad v,x,y\in\R
\end{equation}
in \eqref{eq:identification function}.
The reason for the seemingly more involved representation in \eqref{eq:identification function} including the ``correction term'' $(\id_{\{ y \le v \}} - p)V^*(x,v)$ is the structural similarity to the form of the consistent scoring functions in Theorem \ref{thm:elicitability}; see the discussion in Remark \ref{rem:correction term}.
\end{remark}

\begin{example}
To illustrate Theorem \ref{thm:identifiability}, we let $\rho^*$ be the expectation with a strict $\M^{1}$-identification function $V^*(x,y) = x-y$, $x,y\in\R$. 
The corresponding $p$-tail risk measure $\rho$ induced by $\rho^*$ is $\ES_p$. 
Indeed, a straightforward calculation yields that the second component of \eqref{eq:identification function} is $(1-p)$ times the second component of \eqref{eq:V VaR ES}.
\end{example}

\begin{example}
\label{ex:identifiability RVaR}
  
{  We know that $\RVaR_{p,q}$, $0<p<q<1$, is an $r$-tail risk measure for $r\in(0,p]$. For $r<p$, the generator $\rho^*$ is again an RVaR. For $r=p$, the generator $\rho^*$ is a lower expected shortfall. 
    In both cases, $\rho^*$ fails to be identifiable on reasonably large classes $\M$. 
    As a consequence of Proposition \ref{prop:identifiability}, $(Q_r,\RVaR_{p,q})$, $0<r\le p<q<1$ fails to be identifiable on reasonably large classes.
}
 
\end{example}

\begin{corollary}
\label{cor:CxLS}
Let $(\rho,\rho^*)$ be a $p$-tail pair for some $p\in(0,1)$.
If $\rho^*\colon\M\to\R$ is $\M$-identifiable,
then $(Q_p,\rho)$ satisfies the CxLS* property on $\M_{(p)}\cap\M$.
\end{corollary}

\begin{proof}
Due to Proposition B.4 of the preprint version of \cite{FisslerHoga2023}, Theorem \ref{thm:identifiability} implies that $(Q_p,\rho)$ satisfies the CxLS property on $\M_{(p)}\cap\M$. Since $Q_p$ is $\M_{(p)}$-elicitable, it satisfies the CxLS* property on $\M_{(p)}$. And since $\rho$ can be identified with a singleton-valued functional, $(Q_p,\rho)$ satisfies the CxLS* property on $\M_{(p)}\cap\M$.
\end{proof}

\begin{remark}
\label{rem:equivalence}
Due to Corollary 10 of \cite{SteinwartETAL2014}, under some continuity assumptions on $\rho^*$ and richness assumptions on $\M$, the risk measure $\rho^*$ is $\M$-identifiable if and only if it satisfies the CxLS property on $\M$, which is in turn equivalent to its $\M$-elicitability. 
\end{remark}

\section{Elicitability results}
\label{sec:5}

\subsection{Elicitability relations and score functions}
%\noindent 
Corollary \ref{cor:CxLS} and Remark \ref{rem:equivalence} establish that for a $p$-tail pair $(\rho, \rho^*)$ with an $\M$-elicitable $\rho^*$, the pair $(Q_p,\rho)$ satisfies the CxLS* property on $\M_{(p)}\cap\M$, which is an important necessary condition for the elicitability of $(Q_p,\rho)$.
We emphasize that for multivariate functionals, the CxLS* property fails to be a sufficient condition for elicitability, as shown by \cite{FFHR2021} and \cite{FisslerHoga2023}.
This section first establishes the conditional $\M$-elicitability of $(Q_p,\rho)$ and then provides a simple sufficient condition for its $\M$-elicitability. 
The result below establishes the strictly consistent score function for $\rho^*$ from that of $\rho$ or $(Q_p,\rho)$, if it exists, respectively.

\begin{proposition}
\label{prop:elicitability}
Let $(\rho,\rho^*)$ be a $p$-tail pair for some $p\in(0,1)$. 
%Suppose that $\M$ is convex and contains all degenerate distributions $\delta_y$, $y\in\R$.
The following statements hold.
\begin{enumerate}[(i)]
\item 
For any $r\in\R$ it holds that if 
$\rho$  is $\M\cup\{(1-p)G + p\delta_r : G\in\M\}$-elicitable with a strictly consistent score $S(x,y)$, $x,y\in\R$, then $\rho^*$ is $\M_{\ge r}\cap \M$-elicitable with the strictly consistent score 
\begin{equation}
\label{eq:S*}
S_r^*(x,y) = (1-p)S(x,y) + p S(x,r), \quad x,y\in\R.
\end{equation}
\item
For any $r\in\R$ it holds that if 
$(Q_p, \rho)$ is $\M\cup\{(1-p)G + p\delta_r : G\in\M\}$-elicitable with a strictly consistent score 
$S(v,x,y)$, $v,x,y\in\R$,
then $\rho^*$ is $\M_{\ge r}\cap \M$-elicitable with the strictly consistent score
\begin{equation}
\label{eq:S*2}
S_r^*(x,y) = (1-p) S(r,x,y) + pS(r,x,r), \quad x,y\in\R.
\end{equation}
 \end{enumerate}
\end{proposition} 
The proof is very similar to that  of Proposition \ref{prop:identifiability} and is presented in Appendix \ref{Proof}.

\begin{example}[Continuation of Example \ref{exmp:identifiability}]\label{exmp:elicitability1}
We now illustrate part (i) of Proposition \ref{prop:elicitability}.
A strictly $\M^0$-consistent score for $\VaR_\alpha^-$ is of the form $S(x,y) = \big(\id_{\{y\le x\}} - \alpha\big)\big(g(x) - g(y)\big) $ for a strictly increasing and bounded function $g$ \citep{Gneiting2011b} (upon restricting the class $\M\subseteq \M^0$, we can also choose unbounded $g$). 
For any $p\in(0,\alpha]$ and $r\in\R$
it holds for any random variable $Y$ with distribution $F\in\M_{\ge r}$
that the score $S_r^*$ in \eqref{eq:S*} 
 is almost surely
$$
S^*_r(x,Y) = (1-p)\big(\id_{\{Y\le x\}} - \alpha\big)\big(g(x) - g(Y)\big) + p\big(1 - \alpha\big)\big(g(x) - g(Y)\big), \quad x\in\R.
$$
A direct calculation verifies that this is a positive multiple of 
$$
\big(\id_{\{Y\le x\}} - \alpha^*\big)\big(g(x) - g(Y)\big), \quad \alpha^* = \frac{\alpha - p}{1-p}.
$$
\end{example}

The next result provides  the score function for the tail risk measure $\rho$ from that of $\rho^*$ under some conditions.
A list of technical assumptions adapted from those used by  \cite{FZ16} is needed for   part (iii). Due to its length, this list is presented in Appendix \ref{sec:Assumptions}.

\begin{theorem}
\label{thm:elicitability}
Let $(\rho,\rho^*)$ be a $p$-tail pair for some $p\in(0,1)$.
If $\rho^*\colon\M\to\R$ is $\M$-elicitable with the strictly $\M$-consistent score $S^*\colon \R\times\R\to\R$.
Then the following statements hold.
\begin{enumerate}[(i)]
\item
$\rho$ is $\M$-conditionally elicitable with $Q_p$. In particular, for any $v\in\R$, the score 
\begin{align}
\label{eq:S general 1}
S_v(x,y) = \id_{\{y>v\}} S^*(x,y) +\big(\id_{\{y\le v\}} - p\big) S^*(x,v) +a(y), \quad x,y\in\R,
\end{align}
is strictly $\M(v,Q_p) $-consistent for $\rho$, where $\M(v,Q_p) = \{F\in\M: v\in Q_p(F)\}$.
\item
If for all $x\in\R$, $S^*(x,y)$ is strictly increasing in $y$, then $(Q_p,\rho)$ is $\M$-elicitable with a strictly $\M$-consistent score 
\begin{align}
\label{eq:S general}
S(v,x,y) = \id_{\{y>v\}} S^*(x,y) +\big(\id_{\{y\le v\}} - p\big) S^*(x,v) +a(y), \quad v,x,y\in\R,
\end{align}
 where $a\colon\R\to\R$ is some $\M$-integrable function.
\item
Suppose $\M\subseteq \M_{p}^{\rm c}\cap \M_{(p)}$ and  Assumption \ref{ass:main thm} in Appendix \ref{sec:Assumptions} holds.
Any (strictly) $\M$-consistent score $S\colon\R^3\to\R$ for $(Q_p, \rho)$ is necessarily of the form \eqref{eq:S general}
for almost every $(v,x,y)\in \interior(\A)\times \R$, where $\A = \{(\VaR_p^-(F),\rho(F)) : F\in\M\}\subseteq \R^2$ and
$S^*\colon\R^2\to\R$ is a (strictly) $\M$-consistent score for $\rho^*$ such that  $y\mapsto S^*(x,y)$ is (strictly) increasing for  $x\in \R$.
\end{enumerate}
\end{theorem}
%In part (iii) above, on $\M_{p}^{\rm c}$, $Q_p(F)$ is a singleton and we identify it with its unique element, $\VaR_p^-(F)$.

\begin{proof}
For the part (i), let  $F\in\M$ and $v\in Q_p(F)$. For $S_v$ in \eqref{eq:S general 1}, it holds that 
$$
\int S_v(x,y)\d F(y) = \int S^*(x,y)\d \big(F(y) - p\big)_+ + \int a(y)\d F(y).
$$
Hence
\begin{align*}
\argmin_{x\in\R} \int S_v(x,y)\d F(y) 
= \argmin_{x\in\R} \int S^*(x,y)\d F_p(y) = \rho^*(F_p) = \rho(F),
\end{align*}
which shows the claim. 

For the part (ii), observe that, since $S^*(x,y)$ is strictly increasing in $y$ for all $x\in\R$, 
the function $(v,y)\mapsto S(v,x,y)$ given in \eqref{eq:S general} is a generalized piecewise linear loss, which is strictly $\M$-consistent for $Q_p$, see \cite{Gneiting2011b}.
Hence, for all $x\in\R$,
$$
Q_p(F) = \argmin_{v\in\R} \int S(v,x,y) \d F(y).
$$
The rest follows from part (i).
The proof of part (iii) can be found in Appendix \ref{Proof}. 
\end{proof}

\begin{remark}
\label{rem:correction term}
It should be pointed out that for the elicitability result of Theorem \ref{thm:elicitability} part (ii), we do not need the restriction to $\M_{(p)}\cap \M$ as for the corresponding identifiability result of Theorem \ref{thm:identifiability}.
The reason for this is that, on the one hand, the $p$-quantile $Q_p$ is elicitable on the entire $\M^0$ whereas it is identifiable only on $\M_{(p)}$. 
On the other hand, the presence of the ``correction term'' $(\id_{\{y\le v\}} - p)S^*(x,y)$ in \eqref{eq:S general} is crucial. In fact, part (i) of Theorem \ref{thm:elicitability} would hold on $\M\cap\M_{(p)}$ without this correction term. That is, for any $v\in\R$, the score $S_v(x,y) = \id_{\{y>v\}}S^*(x,y)$ is strictly consistent for $\rho$ on $\big\{F\in\M\cap\M_{(p)}: v\in Q_p(F)\big\}$.
For part (ii), the correction term is also needed to render the score $(v,y) \mapsto  \id_{\{y>v\}}S^*(x,y) +  (\id_{\{y\le v\}} - p)S^*(x,v)$ a strictly $\M$-consistent score for $Q_p$ for any $x$. 
Interestingly, the score $S_v(x,y) = \id_{\{y>v\}}S^*(x,y)$ constitutes a weighted score in the sense of \cite{HolzmannKlar2017} -- see also \cite{GneitingRanjan2011} -- with weight function $w(y) = \id_{\{y>v\}}$.
However, for the score in \eqref{eq:S general}, the weight function obtains a variable threshold, given in terms of a quantile-forecast $v$. 
This extends the existing theory substantially, and it only works with the additional correction term and the monotonicity requirement on $S^*$.
\end{remark}

Recall that if $S^*(x,y)$ is (strictly) $\M$-consistent for $\rho^*$ and if $g\colon\R\to\R$ is $\M$-integrable, then  $\tilde S^* $ given by 
$$
\tilde S^*(x,y) = S^*(x,y) + g(y)
$$
is again (strictly) $\M$-consistent for $\rho^*$. Following \cite{Dawid1998}, we call $S^*$ and $\tilde S^*$ \emph{strongly equivalent}.
Even though they coincide in terms of (strict) consistency for $\rho^*$, they can be clearly different in terms of their monotonicity behavior in $y$. 
The following lemma explores conditions for the existence of a strongly equivalent version of the score which is strictly increasing in its second argument.

\begin{lemma}\label{lemma:sufficient}
Suppose that $S^*\colon\R^2\to\R$ is partially differentiable with respect to its second argument.
Then
there is a differentiable function $g\colon\R\to\R$ such that $y\mapsto S^*(x,y) + g(y)$ is increasing for all $x$ if and only if there is a function $h\colon\R\to\R$ such that 
\begin{equation}
\label{eq:bounded from below}
h(y)\le \partial_y S^*(x,y) \quad \text{ for all }x,y\in\R.
\end{equation}
\end{lemma}
\begin{proof}
If \eqref{eq:bounded from below} holds, then we can choose $g$ to be an antiderivative of $-h$. Moreover, if $g$ is an antiderivative of $-h$ plus a strictly increasing function, then $S^*(x,y) + g(y)$ is even strictly increasing in $y$.
On the other hand, if $ S^*(x,y) + g(y)$ is increasing in $y$, then \eqref{eq:bounded from below} holds with $h=-g'$.
\end{proof}

In practice, the forecast for $Q_p$ might rather play an auxiliary role in comparison to the forecast of the $p$-tail risk measure $\rho$.
The following proposition establishes two relevant notions of \emph{order-sensitivity} for strictly consistent scoring functions of the form \eqref{eq:S general}.
For a general discussion of order-sensitivity, we refer to \cite{FisslerZiegel2019}.

\begin{proposition}
For the strictly $\M$-consistent score in \eqref{eq:S general} and for any $F\in\M$, 
if $v_2<v_1\le \VaR_p^-(F)$ or $\VaR_p^+(F)\le v_1 <v_2$, then 
\begin{enumerate}[\rm (i)]
\item
$
\int S(v_1,x,y)\d F(y) < \int S(v_2,x,y)\d F(y)  \text{ for all }x\in\R;
$
\item
$
\min_{x\in\R}\int S(v_1,x,y)\d F(y) < \min_{x\in\R}\int S(v_2,x,y)\d F(y).
$
\end{enumerate}
\end{proposition}
\begin{proof}
Part (i) follows from the fact that for any fixed $x\in\R$, the score $(v,y)\mapsto S(v,x,y)$ is strictly $\M$-consistent for $Q_p$. This readily implies the claimed order-sensitivity by \cite[Proposition 3]{Nau1985} and \cite[Proposition 3.4]{BB15}.
For part (ii) suppose that $x_1\in \argmin_{x\in\R} \int S(v_1,x,y)\d F(y)$ and $x_2\in\argmin_{x\in\R} \int S(v_2,x,y)\d F(y)$. Then
$\int S(v_2,x_2,y)\d F(y) > \int S(v_1,x_2,y)\d F(y) \ge \int S(v_1,x_1,y)\d F(y)$, where the first inequality is due to part (i).
\end{proof}

\subsection{Special cases and discussions}

We first provide the simple example of the pair of the expectation and ES. 
\begin{example}
Take the $p$-tail pair  $(\rho, \rho^*) = (\ES_p, \E)$. If we choose the strictly $\M^1$-consistent score $S^*(x,y) = \frac{1}{1-p}\big(\phi'(x)(x-y)-\phi(x)\big)+g(y)$ with $\phi$ strictly convex, then the score in \eqref{eq:S general} coincides with the one in \eqref{eq:S ES VaR}.
The condition that $y\mapsto S^*(x,y)$ is strictly increasing is equivalent to the condition that $y\mapsto g(y) - \frac{1}{1-p}\phi'(x)y$ is strictly increasing. 
Invoking Lemma \ref{lemma:sufficient}, a necessary and  sufficient condition for the existence of such a $g$ is that $\phi'$ is bounded from above.
\end{example}

Let us consider the $\tau$-expectile $\rho_\tau$, $\tau\in(0,1)$. On the class $\M^1$ of distributions with finite mean,
\cite{NeweyPowell1987} introduced the $\tau$-expectile as the unique solution $x=\rho_\tau(F)$ to the equation
$$
\tau\int_x^\infty (y-x)\d F(y) = (1-\tau)\int_{-\infty}^x(x-y)\d F(y).
$$
For $\tau=1/2$, this family contains the usual expectation.
\cite{BelliniETAL2014} showed that for $\tau\ge1/2$, the $\tau$-expectile constitutes a coherent risk measure.   \cite{Z16} established that these are the only elicitable law-based coherent risk measures.
The following elicitability result about induced $p$-tail risk measures is novel to the literature.

\begin{proposition}\label{prop:expectile}
Let $(\rho,\rho^*)$ be a $p$-tail pair and 
$\rho^*$ the $\tau$-expectile, $\tau\in(0,1)$. 
The pair $(Q_p, \rho)$ is $\M^1$-elicitable.
\end{proposition}
\begin{proof}
Subject to mild regularity and integrability conditions, any strictly $\M^1$-consistent score for $\rho^*$ is of the form
$$
S^*(x,y) = |\id_{\{y\le x\}} - \tau|\big(\phi(y)-\phi(x) + \phi'(x)(x-y)\big) + g(y),
$$
where $\phi$ is strictly convex with subgradient $\phi'$ and $g$ is arbitrary \citep[Theorem 10]{G11}. 
For $g=0$ and $y\neq x$, the partial derivative is 
$$
\partial_y S^*(x,y) = |\id_{\{y\le x\}} - \tau|\big(\phi'(y)-\phi'(x)\big).
$$
For $y=x$, the left-sided and right-sided partial derivative with respect to $y$ exist, coincide and are both 0.
Take $\phi$ such that $|\phi'|<C$ for some $C>0$, for example we can use $\phi(x) = x^2/(1+|x|)$, then $|\partial_y S^*(x,y) |<2C$.
Therefore, we can apply Theorem \ref{thm:elicitability} and Lemma \ref{lemma:sufficient} to construct a strictly $\M^1$-consistent score for the $p$-tail-$\tau$-expectile and $Q_p$.
\end{proof}

\begin{proposition}\label{prop:ratio}
Let $\rho^*\colon\M\to\R$ be the ratio of expectations, that is,
$$
\rho^*(F) = \frac{\int u(y)\d F(y)}{\int t(y)\d F(y)}
$$
for two functions $u,t\colon\R\to\R$, where $\M$ is chosen such that $\rho^*$ is well-defined and finite.
Let $(\rho,\rho^*)$ be a $p$-tail pair. If $u$ and $t$ are differentiable, 
then the pair $(Q_p, \rho)$ is $ \M$-elicitable.
\end{proposition}

\begin{proof}
Subject to mild regularity and integrability conditions, any strictly $\M$-consistent score for $\rho^*$ is of the form
\begin{align*}
S^*(x,y) &= -\phi(x)t(y) + \phi'(x)\big(xt(y) - u(y)\big)+ g(y)\\
&= t(y)\big(\phi'(x)x - \phi(x)\big) - t(y)\phi'(x) + g(y).
\end{align*}
where $\phi$ is strictly convex with subgradient $\phi'$ and $g$ is arbitrary \citep[Theorem 8]{G11}. 
For the strictly convex function $\phi(x) = x^2/(1+|x|)$, we obtain that $\phi'(x) \in[-1,1]$ and $\phi'(x)x - \phi(x)\in[0,1]$ for all $x\in\R$.
If $u$ and $t$ are differentiable, then $|\partial_y S^*(x,y)|\le |u'(y)|+|t'(y)|$ such that we can apply Lemma \ref{lemma:sufficient} and Theorem \ref{thm:elicitability}.  
\end{proof}

\begin{remark}
We can further relax the differentiability condition on $u$ and $t$ in Proposition \ref{prop:ratio}.
Suppose that 
both $u$ and $t$ have a finite total variation on any compact interval, and we use the notation
$$
\|u\|(y) = \mathrm{sign}(y)\sup_{P\in\mathcal P(y)} \sum_{i=0}^{n_{P}-1}|u(z_{i+1})-u(z_i)|, \qquad y\in\R,
$$
where $\mathcal P(y)$ denotes the set of all partitions $P=\{z_0,z_1,\ldots,z_{n_P} : z_i<z_{i+1}\}$ of the interval $[\min(0,y),\max(0,y)]$, then 
 we obtain a strictly consistent score $S$ in \eqref{eq:S general} upon choosing 
$g(y) = \|u\|(y) + \|t\|(y) +y$, $y\in\R$. 
For example, if $u= 1$ and $t(y) = y^2$, then $\|u\|(y)=0$ and $\|t\|(y) = \mathrm{sign}(y)y^2$ for $y\in\R$. If $u(y) = \id_{[a,b)}(y)$ for $a<b$, then $\|u\|(y) = \id_{[a,\infty)}(y)+\id_{[b,\infty)}(y)$ for $y\in\R$.
\end{remark}

\begin{example}\label{exmp:shortfall risk}
Consider $\rho^*$ to be a shortfall risk measure induced by the loss function $\ell \colon\R\to\R$, which is left-continuous, non-decreasing and $\inf_{x\in\R}\ell(x)<0<\sup_{x\in\R}\ell(x)$. Then
$$
\rho^*(F) = \inf\left\{m\in\R: \int \ell(y-m)\d F(y)\le0\right\}.
$$
Further, suppose that we consider a class $\M$ of distributions such that for all $F\in\M$ and for all $m\in\R$
$$
m=\rho^*(F) \quad \text{ if and only if }\quad \int \ell(y-m)\d F(y)=0.
$$
This means that $\rho^*$ is $\M$-identifiable with the strict $\M$-identification function $V^*(m,y) = \ell(y-m)$. We can apply Osband's principle \citep[Theorem 3.2 and Corollary 3.3]{FZ16} to see that -- subject to regularity conditions -- any strictly $\M$-consistent score $S^*$ for $\rho^*$ is of the form
%\begin{equation}
%\label{eq:S shortfall}
$$
S^*(x,y) = \int_0^x h(z)\ell(y-z)\d z + g(y)
$$
%\end{equation}
for some non-positive function $h\colon\R\to\R$. The question as to what choices of $h$ and $g$ lead to a score which is strictly increasing in its second argument is hard to answer in full generality. We only provide a sufficient condition: 
If $\ell$ is bounded from below, then we can find some function $g$ such that 
$$
S^*(x,y) = - \int_0^x \ell(y-z)\d z + g(y)
$$
is increasing. Indeed, 
$\partial_y \big[S^*(x,y) -g(y)\big] = \ell(y-x) - \ell(y)$. Hence, by Lemma \ref{lemma:sufficient} and Theorem \ref{thm:elicitability}, the pair $(Q_p,\rho)$ is $\M$-elicitable.
We emphasize that the boundedness of $\ell$ is not necessary for the elicitability of $(Q_p,\rho)$ as the case of $\rho^* = \E$ with $\ell(x) = x$ shows (or more generally the $\tau$-expectile; see Proposition \ref{prop:expectile}). 
\end{example}

\begin{remark}
{The   CoVaR \citep{AdrianBrunnermeier2016}  has become an important systemic risk measure. In its definition according to \cite{GT2013} and \cite{NoldeZhang2020}, it closely resembles a $p$-tail risk measure. For $\alpha\in(0,1)$ and $\beta\in[0,1)$, the CoVaR$_{\alpha|\beta}$ of a two-dimensional random vector $(X,Y)$ with joint distribution $F_{X,Y}$ is defined as CoVaR$_{\alpha|\beta}(F_{X,Y}) = \VaR_{\alpha}^-(F_{Y|X\ge\VaR_\beta^-(F_X)})$, where $F_{Y|X\ge\VaR_\beta^-(F_X)}$ is the conditional distribution of $Y$ on $\{X\ge\VaR_\beta^-(F_X)\}$. We assume that the marginal distributions are continuous and strictly increasing such that the distinction between the lower and upper quantile is inessential.
\cite{FisslerHoga2023} showed that  $F_{X,Y} \mapsto (\mathrm{CoVaR}_{\alpha|\beta}(F_{X,Y}), \VaR_\beta(F_X))$ generally fails to be elicitable (even though it is identifiable and thus has convex level sets), where $F_X$ is  the distribution of $X$. 
The reason why it is not possible to leverage the construction principle of Theorem \ref{thm:elicitability} is the fact that for CoVaR, the observation process is bivariate.}
\end{remark}

\begin{example}
    
{Similarly to the discussion in Example \ref{ex:identifiability RVaR}, Proposition \ref{prop:elicitability} implies that on reasonably large classes $\M$, the pair $(Q_r, \RVaR_{p,q})$, $0<r\le p<q <1$, fails to be elicitable.}
\end{example}

\subsection{Elicitation of the tail distribution}

In a risk management context, but also in other statistical contexts, one could be interested in the \emph{entire} tail distribution $F_p$ for some $F\in\M^0$, $p\in(0,1)$, instead of a given tail risk measure. 

Consistent scoring functions are tailored to evaluate the accuracy of a point forecast for a certain functional of interest. 
Their counterpart for probabilisitic forecasts, specifying the entire distribution, are \emph{proper scoring rules} \citep{GneitingRaftery2007}.
Mathematically, a scoring rule is a map $s\colon\M\times \R\to\R\cup\{\infty\}$ for some $\M\subseteq \M^0$. It is \emph{$\M$-proper} if  for all $F\in\M$
$$
F\in \argmin_{G\in\M} \int s(G,y)\d F(y).
$$
It is \emph{strictly $\M$-proper} if $F$ is the unique element of the above argmin.
In a similar vein, we define that $s$ is $\M$-proper for the $p$-tail if for all $F\in\M$
$$
\{H\in\M: H_p = F_p\}\subseteq \argmin_{G\in\M} \int s(G,y)\d F(y).
$$
Similarly, $s$ is strictly $\M$-proper for the $p$-tail if we have equality above.
Tail scores have been hinted at in the discussion at the end of Subsection 2.3 in \cite{HolzmannKlar2017}.

The most prominent proper scoring rules are the log-score, $\mathrm{LogScore}(f,y) = -\log(f(y))$, where $f$ is a predictive density, and the continuous ranked probability score (CRPS), 
$$\mathrm{CRPS}(F,y) = \int \big(F(z) - \id_{\{y\le z\}}\big)^2\d z. 
$$
The CRPS is strictly $\M^1$-proper.
Moreover, one can check that $\partial_y \mathrm{CRPS}(F,y) = 2F(y) -1$ for all $F\in\M^1$. 
Hence, invoking Lemma \ref{lemma:sufficient} and using a similar construction as in Theorem \ref{thm:elicitability} (ii), the map
 $$
 S(v,F,y) = \id_{\{y>v\}}\big(\mathrm{CRPS}(F,y)+2y\big) + \big(\id_{\{y\le z\}} - p\big) \big(\mathrm{CRPS}(F,v)+2v\big), \quad v,y\in\R, \ F\in\M^1
 $$
 is strictly $\M$-proper for 
$F\mapsto(Q_p(F), F_p)$.

While, to the best of our knowledge, the above score is novel to the statistical literature, an alternative strictly $\M^1$-proper score for the $p$-tail has been known in the form of the quantile weighted CRPS \citep{GneitingRanjan2011}
$$
\mathrm{CRPS}_p(F,y) = \int_p^1 2\big(\id_{\{y\le F^{-1}(r)\}} - r\big)\big(F^{-1}(r) - r\big)\d r.
$$ 
This result hinges on the fact that for any $F\in\M^0$ and $p\in(0,1)$, we have a bijection between $F_p$ and $\big(F^{-1}(r)\big)_{r\in(p,1]}$.

\section{Results on the left tail  or the body of the distribution}
\label{sec:sign conventions}

We have been working with tail risk measures defined for the right tail distribution.
Mathematically, it is also possible to translate all our results to the left tail distribution.
This may be useful when the positions of gains and losses are switched, as in, e.g., the convention of \cite{FS16}, or when the best scenarios of random outcomes are of interest. 
We only mention a few useful points when converting to the left tail.

For $q\in(0,1]$, we can define the left tail distribution of $F\in\M^0$ as 
$$
F^q(x) = \frac{\min(F(x),q)}{q},  \quad x\in\R.
$$ 
To formulate the corresponding tail risk measure, 
for a risk measure $\rho^*\colon\M\to\R$ and a class $\M$ such that $F^q\in\M$ for each $F\in\M$, we can define 
%\begin{equation}
%\label{eq:lower tail risk}
$$\rho^q(F) = \rho^*(F^q).$$
%\end{equation} 
The pair $(\rho^q,\rho^*)$ is called \emph{left $q$-tail pair}.
With the above definitions, 
we can reproduce our main results as well. We only state the most interesting and relevant counterpart, namely the one of Theorem \ref{thm:elicitability} part (ii). 

\begin{theorem}\label{thm:elicitability left tail risk}
Let $(\rho^q, \rho^*)$ be a left $q$-tail pair for some $q\in(0,1)$. If $\rho^*\colon\M\to\R$ is $\M$-elicitable with the strictly $\M$-consistent score $S^*\colon\R\times\R\to\R$ and if for all $x\in\R$, $S^*(x,y)$ is strictly \emph{decreasing} in $y$, then
\begin{align*}
%\label{eq:S lower tail risk}
S(v,x,y) &= \id_{\{y\le v\}} S^*(x,y) - \big(\id_{\{y\le v\}}-q\big)S(x,v) + a(y), \\ \nonumber
 &= -\id_{\{y> v\}} S^*(x,y) - \big(\id_{\{y\le v\}}-q\big)S(x,v) + \tilde a(y),	\quad v,x,y\in\R,
\end{align*}
is strictly $\M$-consistent for $(Q_q,\rho^q)$, where $a\colon\R\to\R$ and $\tilde a\colon\R\to\R$ are some $\M$-integrable functions.
\end{theorem}

The proof of Theorem \ref{thm:elicitability left tail risk} is almost identical to the one of Theorem \ref{thm:elicitability} (ii) and therefore omitted. 
While the left-tail results can be obtained directly form the right-tail results, 
more work is needed to combine these two tails. One relevant question is related to the body part of the distribution,
which can be seen as the intersection of two tails. 
For $0\le p<q\le 1$ and $F\in\M^0$, we can define $F^{[p,q]}\in\M^0$ by 
$$
F^{[p,q]}(x) = \frac{(\min(F(x),q) - p)_+}{q-p}, \quad x\in\R,
$$ and a risk measure $\rho^{[p,q]}$ via another generating risk measure $\rho^*$ by
\begin{equation}
\label{eq:rho p q}
    \rho^{[p,q]}(F) = \rho^*(F^{[p,q]}).
\end{equation}

{The most prominent example of a risk measure depending on the body of the distribution is RVaR. In particular, for $0<p<q<1$, $\RVaR_{p,q}$ arises in \eqref{eq:rho p q} with $\rho^*$ as the mean.}

We obtain the following counterpart to Theorem \ref{thm:elicitability} part (ii) and Theorem \ref{thm:elicitability left tail risk}.

\begin{theorem}\label{thm:elicitability rho p q}
Let $\rho^*\colon\M\to\R$ be $\M$-elicitable with strictly $\M$-consistent score $S^*\colon\R\times\R\to\R$.
For $0<p<q<1$, let $\rho^{[p,q]}$ be defined via \eqref{eq:rho p q}. 
Then the function given by
\begin{align}
\label{eq:score p q}
&S(v_1,v_2,x,y) \\ \nonumber
&= \id_{\{v_1<y \le v_2\}}S^*(x,y) +\big(\id_{\{y\le v_1\}}-p\big)S^*(x,v_1) - \big(\id_{\{y\le v_2\}}-q\big)S^*(x,v_2) \\
& \quad + \big(\id_{\{y\le v_1\}}-p\big)g_1(v_1) + \id_{\{y>v_1\}}g_1(y)
+ \big(\id_{\{y\le v_2\}}-p\big)g_2(v_2) + \id_{\{y>v_2\}}g_2(y) + a(y),\nonumber
\end{align}
is strictly $\M$-consistent for $(Q_p,Q_q,\rho^{[p,q]})$, where $a\colon\R\to\R$ and $\tilde a\colon\R\to\R$ are some $\M$-integrable functions, and $g_1\colon\R\to\R$ and $g_2\colon\R\to\R$ are such that 
for each $x\in\R$,
the functions
	\begin{equation}
    \label{eq:monotonicity condition}
	    v_1 \mapsto g_1(v_1) + S^*(x,v_1) \quad \text{and }\quad 
	v_2 \mapsto g_2(v_2)  - S^*(x,v_2)
	\end{equation}
are strictly increasing. 
\end{theorem} 
\begin{proof}
    Due to the monotonicity conditions in \eqref{eq:monotonicity condition}, for fixed $v_2,x\in\R$, the score $(v_1,y)\mapsto S(v_1,v_2,x,y)$ 
    is a generalized piecewise linear loss, and hence $\M$-strictly consistent for $Q_p$.
    Similarly, for fixed $v_1,x\in\R$, the score $(v_2,y)\mapsto S(v_1,v_2,x,y)$ is strictly $\M$-consistent for $Q_q$. Finally, for $F\in\M$ and $v_1\in Q_p(F)$, $v_2\in Q_q(F)$ it holds that \begin{align*}
        \int S(v_1,v_2,x,y)\d F(y) 
        = \int S^*(x,y) \d \big(\min(F(x),q)-p\big)_+ + \int \tilde a(y)\d F(y),
    \end{align*} 
    where $\tilde a$ does not depend on $x$. Hence,
    \begin{align*}
        \argmin_{x\in\R}\int S(v_1,v_2,x,y)\d F(y) 
        = \argmin_{x\in\R}\int S^*(x,y) \d F^{[p,q]}(x) = \rho^*(F^{[p,q]}) = \rho^{[p,q]}(F),
    \end{align*}
    which shows the claim.
\end{proof}
 
Identifiability results in the spirit of Theorem \ref{thm:identifiability} can also be obtained. 

\begin{theorem}
\label{thm:identifiability rho p q}

{Let $\rho^*\colon\M\to\R$ be $\M$-identifiable with strict $\M$-identification function $V^*\colon\R\times\R\to\R$.
For $0<p<q<1$, let $\rho^{[p,q]}$ be defined via \eqref{eq:rho p q}. 
Then $(Q_p, Q_q, \rho^{[p,q]})$ is $\M_{(p)}\cap \M_{(q)} \cap \M$-identifiable with a strict $\M_{(p)}\cap \M_{(q)} \cap \M$-identification function $V:\R^4\to \R$ given by
\begin{equation}
\label{eq:identification function 2}
\begin{pmatrix}
 \id_{\{ y \le v_1 \}} - p \\
 \id_{\{ y \le v_2 \}} - q \\
\id_{\{v_1<y \le v_2\}}V^*(x,y) +\big(\id_{\{y\le v_1\}}-p\big)V^*(x,v_1) - \big(\id_{\{y\le v_2\}}-q\big)V^*(x,v_2)
 \end{pmatrix}
\end{equation}
 for $ (v_1,v_2,x,y )\in\R^4 $.}
\end{theorem}
{The proof of Theorem \ref{thm:identifiability rho p q}
is completely analogous to the one of Theorem \ref{thm:identifiability} and therefore omitted. 
In line with Remark \ref{rem:identifiability}, instead of the identification function in \eqref{eq:identification function 2} we can also use the strict $\M_{(p)}\cap \M_{(q)} \cap \M$-identification function without ``correction terms'' in the form of 
\begin{equation*}
V(v_1,v_2,x,y) = \begin{pmatrix}
 \id_{\{ y \le v_1 \}} - p \\
 \id_{\{ y \le v_2 \}} - q \\
\id_{\{v_1<y \le v_2\}}V^*(x,y) 
 \end{pmatrix}.
\end{equation*}  
    Theorems \ref{thm:elicitability rho p q} and \ref{thm:identifiability rho p q} generalize results on the elicitability and identifiability of the triplet $(Q_p,Q_q,\RVaR_{p,q})$, $0<p<q<1$, from \cite{FZ21}, which are summarized in Example \ref{ex:RVaR}.
    In fact, using the canonical identification function $V^*(x,y) = x-y$ for the mean $\rho^*$, the identification function in \eqref{eq:identification function 2} equals the one in \eqref{eq:V RVaR}, except for a factor of $q-p$ in the third component.
    Similarly, the scoring function in \eqref{eq:score p q} yields the strictly consistent score of Equation \eqref{eq:S RVaR} upon choosing the strictly consistent score for the mean $S^*(x,y) = \frac{1}{q-p}\big(\phi'(x)(x-y) - \phi(x)\big)$, where $\phi$ is strictly convex. 
    Also the monotonicity conditions in \eqref{eq:monotonicity condition} are equivalent to 
    $v_1\mapsto g_1(v_1) - \frac{1}{q-p}\phi'(x)v_1$ 
and 
$v_2\mapsto g_2(v_2) + \frac{1}{q-p}\phi'(x)v_2$
being strictly increasing for all $x\in\R$.
}

 \section{Concluding remarks}\label{sec:6} 

Our main results in Sections \ref{sec:4} and \ref{sec:5} establish an intimate connection between identifiability or elicitability of a tail risk measure $\rho$ and that of its generator $\rho^*$, as well as their identification functions and score functions. 
  To summarize briefly,
  under suitable conditions,
the corresponding property of $\rho^*$ implies that of $(Q_p,\rho)$, 
and the corresponding property of $\rho$ or $(Q_p,\rho)$ implies that of $\rho^*$.  
Moreover, there is an explicit way to convert between their identification functions and  between their score functions.

Some open questions remain. 
First, elicitability of risk measures defined for the intermediate region of the distribution, instead of the tail distributions, has been discussed briefly 
in Section \ref{sec:sign conventions}, without being fully developed, as we only obtained a parallel result to Theorem \ref{thm:elicitability} part (ii). 

Second, we provided in Theorem \ref{thm:elicitability} some sufficient conditions on the elicitable risk measure $\rho^*$ that guarantee 
   the elicitability of 
   the $p$-tail risk measure $\rho$ itself or the pair $(Q_p, \rho)$. 
A full characterization of all elicitable risk measures $\rho^*$ yielding the  elicitability of $\rho$ or $(Q_p, \rho)$  is missing. 
 Theorem \ref{thm:elicitability}  does not provide a full answer to this question. In particular, one needs to rely on the knowledge of the existence of a strictly consistent scoring function for $\rho^*$, $S^*(x,y)$, that is strictly increasing in $y$.  
Third, the elicitation complexity (\citealp{FK21}) of tail risk measures remains unclear. 
Roughly (with suitable regularizing conditions), the elicitation complexity of $\rho$ is a positive integer $k$ that quantifies how many dimensions are needed to make $\rho$ a function of an elicitable  vector $(\rho_1,\dots,\rho_k)$. 
It may be tempting to 
conjecture that if $\rho^*$ has elicitation complexity $k$, then $\rho$ has elicitation complexity at most $k+1$, seeing from the example of $(\rho,\rho^*)=(\ES_p,\E)$ with elicitation complexity $(2,1)$; see \cite{FK21}. This question seems to be very challenging to answer with current techniques.

   \subsection*{Conflict of interest statement}
{The authors declare no conflict of interest.}

   \subsection*{Data availability statement}
No datasets were generated or analyzed during the current study. 
  
   \subsection*{Acknowledgements}
   %\noindent 

   The authors thank the Editor, an Associate Editor, and the anonymous referees for helpful comments. 
 F.~Liu  and R.~Wang  acknowledge  financial support from the Natural Sciences and Engineering Research Council of Canada (RGPIN-2024-03728, CRC-2022-00141,  RGPIN-2020-04717, DGECR-2020-00340).
 L.~Wei acknowledges financial support from the National Natural Science Foundation of China (No.~12271415).

 \appendix 

\section*{Appendices}

\section{Risk measures}
\label{app:Risk measures}

%All random variables are defined on an atomless probability space $(\Omega, \F, \p)$. 
{We take the mapping $\rho:\M\to [-\infty,\infty]$ as the primary object in the paper. In the literature, a  risk measure are often defined as a mapping $\hat \rho$ from a set of random variables $\mathcal X$ on an atomless probability space $(\Omega, \mathcal F, \mathbb P)$ to $[-\infty,\infty]$ or $\R$ \citep{FS16}. 
Commonly used $\hat \rho $ are 
 \emph{law-based} (also called \emph{law-invariant}),  meaning that for any $X,X'\in\mathcal X$ it holds that $\hat \rho(X) = \hat \rho(X')$ whenever the distributions of $X$ and $X'$ coincide.
 It is easy to see that 
a  law-based  $\hat \rho:\X\to [-\infty,\infty]$ one-to-one corresponds  to 
 a risk measure $\rho:\M\to [-\infty,\infty]$ with $\M = \{F_X\in\M^0 : X\in\mathcal X\}$  via
\begin{equation}\label{eq:def rho}
 \rho(F)  = \hat \rho(X) \quad \text{ where $X$ is distributed as $F$.}
\end{equation} 
 Invoking the relation \eqref{eq:def rho}, our study can easily be translated into properties of $\hat \rho$.
Here, we follow the sign convention of \cite{QRM15} to consider random variables in $\X$ as losses.} 
 
%\noindent 
Some commonly used desirable properties of {risk measure}s are listed below. 
\begin{enumerate}[(A4)]
\item[(A0)] \textit{Law-invariance}: if $X,Y\in \X$ and $X\laweq Y$, then $\hat \rho(X)=\hat \rho(Y)$.
\item[(A1)] \textit{Monotonicity}: $\hat \rho(X)\leq \hat \rho(Y)$ if $X \leq Y $ a.s, $ X,Y \in \X$.
 \item[(A2)] \textit{Translation-equivariance}:  $\hat \rho(X-m)=\hat \rho(X)-m$ for $m\in\R$ and $X\in \X$.
\item[(A3)] \textit{Convexity}:  $\hat \rho(\lambda X+(1-\lambda)Y)\leq \lambda\hat \rho(X)+(1-\lambda)\hat \rho(Y)$ for  $\lambda\in[0,1]$ and $X,Y\in \X$.
  \item[(A4)] \textit{Positive homogeneity}:  $\hat \rho(\lambda X)=\lambda \hat \rho(X)$ for $\lambda> 0$ and $X\in \X$.
 \item[(A5)] \textit{Subadditivity}: $\hat \rho(X+Y)\leq \hat \rho(X)+\hat \rho(Y)$ for $X,Y\in \X$.
  \item[(A6)] \textit{Comonotonic additivity}:  $\hat \rho(X+Y)=\hat \rho(X)+\hat \rho(Y)$ if $X, Y\in \X$ are comonotonic.\footnote{Two random variables $X$ and $Y$ are \emph{comonotonic} if there exists $\Omega_0\in \mathcal F$ with $\p(\Omega_0)=1$ and $  (X(\omega)-X(\omega'))(Y(\omega)-Y(\omega'))\geq 0~~\mbox{for all}~\omega,\omega'\in \Omega_0.
 $ Comonotonicity of  $X$  and $Y$ is equivalent to the existence of a random variable $Z\in L^0$ and two non-decreasing functions $f$ and $g$ such that $X=f(Z)$ and $Y=g(Z)$ almost surely. We refer to \cite{DDGKV02} for an overview on comonotonicity.}
   \end{enumerate}
 
It is well known that any pair of properties (A3), (A4) and (A5) implies the remaining third one.
For economic interpretations of these properties, we refer to  \cite{ADEH99}, \cite{FS16} and \cite{D12}.

\begin{definition} \label{def:risk measures}
A risk measure $\rho$ is 
a  \emph{monetary risk measure} if its corresponding $\hat\rho$ satisfies (A1) and (A2),  it is a \emph{convex risk measure} if $\hat\rho$ satisfies (A1)--(A3), and it is  a \emph{coherent risk measure} if $\hat\rho$ satisfies  (A1)--(A4). 
%The \emph{acceptance set} of a monetary risk measure $\rho$ is defined as $\mathcal A_\rho=\{X\in \X:\hat \rho(X)\le 0\}.$
\end{definition}
{Some of the above terminologies appear in the main paper, but they are not essential for our results.}

\section{Assumptions}
\label{sec:Assumptions}

%\noindent 
The following list of assumptions is an adaptation of the assumptions for Proposition 1 in \cite{erratum}; see also \cite{FZ16}. 
This list of assumptions is needed for  Theorem \ref{thm:elicitability}, part (iii). 
For an identification and scoring functions $V(x,y)$ and $S(x,y)$, $x\in\R^k$, $y\in\R$, we use the shorthands 
$\bar V(x,F):= \int V(x,y)\d F(y)$ and $\bar S(x,F):= \int S(x,y)\d F(y)$ for $F\in\M^0$, tacitly assuming that the integral is well-defined.
In the sequel, we shall only consider $\M\subseteq \M_p^{\rm c}\cap\M_{(p)}$. On $\M_{p}^{\rm c}$, $Q_p(F)$ is a singleton and we identify it with its unique element, $\VaR_p^-(F)$.
Moreover, here and in the proof of Section \ref{Proof}, we work with an identification function of the form \eqref{eq:identification function2}. 
We remark that we could equivalently work with an identification function of the form \eqref{eq:identification function}.
\begin{assumption}\label{ass:main thm}
\begin{enumerate}[(i)]
\item
Let $\A = \{(\VaR_p^-(F),\rho(F)) : F\in\M\}\subseteq \R^2$ and suppose that the interior of $\A$, $\interior(\A)$ is simply connected.
\item
The generator $\rho^*$ is identifiable with a strict $\M$-identification function $V_{\rho^*}\colon \R^2\to\R$, which is locally bounded.
\item
The identification function $V\colon\R^3\to\R^2$ defined by 
$$
V(v,x,y) = \begin{pmatrix}
\id_{\{y\le v\}} - p \\
\id_{\{y>v\}} V_{\rho^*}(x,y)
\end{pmatrix}
$$
is such that its expectation with respect to any $F\in\M$
is continuously differentiable. In particular, any $F\in\M$ has a continuous derivative $f$ (which is also its Lebesgue density).
\item
The expectation of the score $S$ with respect to any $F\in\M$ is twice continuously differentiable.
\item
Suppose that $\M$ is convex and that for any 
 $(v,x)\in\A$ there  are $F_1,F_2,F_3 \in\M$ such that 0 is contained in the interior of the convex hull of $\{\bar V(v,x,F_1), \bar V(v,x,F_2), \bar V(v,x,F_3)\}$.
 \item
 Suppose that for any $(v,x)\in\A$ there are $F_1,F_2\in\M$ with derivatives $f_1$ and $f_2$ such that $(v,x) = (\VaR_p^-(F_1),\rho(F_1))=(\VaR_p^-(F_2),\rho(F_2))$, $\partial_x \bar V_2(v,x,F_1) = \partial_x \bar V_2(v,x,F_2)$, but $f_1(v)\neq f_2(v)$.
\item
Suppose that the complement of the set 
$$
\{(v,x,y)\in \A\times \R : V(v,x,\cdot) \text{ and } S(v,x,\cdot) \text{ are continuous at the point }y\}
$$
has 3-dimensional Lebesgue measure zero.
\item
For every $y\in\R$ there exists a sequence $(F_n)_{n\in\mathbb N}$ of distributions $F_n\in\M$ that converges weakly to the Dirac-measure $\delta_y$ such there is some $C>0$ such that $F_n([-C,C])=1$ for all $n\in\mathbb N$.
\end{enumerate}
\end{assumption}

\section{Proofs in Section \ref{sec:5}}
\label{Proof}

\begin{proof}[Proof of Proposition \ref{prop:elicitability}] 
\begin{enumerate}[(i)]
\item
Let $S$
be a strictly $\M\cup\{(1-p)G + p\delta_r : G\in\M\}$-consistent score for $\rho$ and $S_r^*$ as given in \eqref{eq:S*} for some $r\in\R$.
Choose some $G\in\M_{\ge r} \cap \M$ and define $F(y) = \id_{\{y\ge r\}}\big((1-p)G(y)+p\big)$.
Due to the assumed elicitability of $\rho$ and since $F_p = G$, we get
\begin{align*}
\rho^*(G) = \rho^*(F_p) = \rho(F) 
&= \argmin_{x\in\R}  \int S(x,y) \d F(y) \\
&= \argmin_{x\in\R} \int S(x,y) \d \left (\id_{\{y\ge r\}}(1-p)G(y)\right ) + pS(x,r)\\
&= \argmin_{x\in\R}  \int \big ((1-p) S(x,y) + pS(x,r)\big )\d G(y)\\
&=\argmin_{x\in\R}  \int S^*_r(x,y)\d G(y) ,
\end{align*}
which shows the claim.

\item
Let $S(v,x,y)$, $v,x,y\in\R$, be a strictly $\M\cup\{(1-p)G + p\delta_r : G\in\M\}$-consistent score and $S_r^*$ given in \eqref{eq:S*2} for some $r\in\R$.
For any $F\in \M\cup\{(1-p)G + p\delta_r : G\in\M\}$ and $v\in Q_p(F)$ it holds that 
$$
\rho(F) = \argmin_{x\in\R} \int_{\R} S(v,x,y)\d F(y).
$$
Let $G\in\M_{\ge r}\cap \M$. Then, as in part (i), define $F(y) = \id_{\{y\ge r\}}\big((1-p)G(y) + p\big)$, resulting in $F\in \M\cup\{(1-p)G + p\delta_r : G\in\M\}$ and $r = \VaR_p^-(F)\in Q_p(F)$. 
Again, $F_p = G$, which implies 
\begin{align*}
\rho^*(G) = \rho^*(F_p) = \rho(F) 
&=  \argmin_{x\in\R}  \int S(r,x,y) \d F(y) \\
&= \argmin_{x\in\R}  \int S(r,x,y) \d \left (\id_{\{x\ge r\}}(1-p)G(x) \right )+ pS(r,x,r) \\
&= \argmin_{x\in\R} \int \left ((1-p) S(r,x,y) + pS(r,x,r)\right )\d G(x)\\
&=\argmin_{x\in\R}  \int S^*_r(x,y)\d G(y),
\end{align*}
which shows the claim. \qedhere 
\end{enumerate}
\end{proof}

For a proof of Theorem \ref{thm:elicitability} part (iii), in the sequel, we shall only consider $\M\subseteq \M_p^{\rm c}\cap\M_{(p)}$. On $\M_{p}^{\rm c}$, $Q_p(F)$ is a singleton and we identify it with its unique element, $\VaR_p^-(F)$.
Moreover, here and in the Assumption Section \ref{sec:Assumptions}, we work with an identification function of the form \eqref{eq:identification function2}. 
We remark that we could equivalently work with an identification function of the form \eqref{eq:identification function}.
\begin{proof}[Proof of Theorem \ref{thm:elicitability} part (iii)]
Using Theorem \ref{thm:identifiability} and Remark \ref{rem:identifiability}, the function
$$
V(v,x,y) = \begin{pmatrix}
\id_{\{y\le v\}} - p \\
\id_{\{y>v\}} V_{\rho^*}(x,y)
\end{pmatrix}, \quad v,x,y\in\R
$$
is a strict $\M$-identification function for $(\VaR_p^-,\rho)$; see also \cite[Proposition 3.4]{FisslerHoga2023}.
We obtain for $F\in\M$ with derivative (and density) $f$
\begin{align*}
\bar V_1(v,x,F) &= F(v) - p,
\qquad \partial_v \bar V_1(v,x,F) = f(v),
\qquad \partial_x \bar V_1(v,x,F) =0,\\
\bar V_2(v,x,F) &= \int_v^\infty V_{\rho^*}(x,y)\d F(y), 
\qquad \partial_v \bar V_2(v,x,F) = -V_{\rho^*}(x,v)f(v).
\end{align*}
Due to Osband's principle \citet[Theorem 3.2]{FZ16}, there is a matrix-valued function $h\colon\interior(\A)\to\R^{2\times2}$ such that 
$$
\nabla \bar S(v,x,F) = h(v,x) \nabla V(v,x,F) \qquad \text{for all }(v,x)\in\interior(\A), \ F\in\M.
$$
We obtain for the second order derivatives
\begin{align*}
\partial_x\partial_v  \bar S(v,x,F)
&= \partial_x h_{11}(v,x) \big(F(v) - p\big) + h_{12}(v,x) \partial_x \bar V_2(v,x,F) + \partial_x h_{12}(v,x) \bar V_2(v,x,F),\\
\partial_v\partial_x  \bar S(v,x,F)
&=  h_{21}(v,x)f(v) + \partial_v h_{21}(v,x) \big(F(v) - p\big) - h_{22}(v,x) V_{\rho^*}(x,v)f(v) \\& \qquad + \partial_v h_{22}(v,x) \bar V_2(v,x,F).
\end{align*}
The Hessian of the expected score $\bar S(v,x,F)$ needs to be symmetric for all $(v,x)\in\interior(\A)$ and for all $F\in\M$.
Evaluating $\partial_x\partial_v  \bar S(v,x,F) = \partial_v\partial_x  \bar S(v,x,F)$ for $(v,x) = (\VaR_p^-(F),\rho(F))$, we get
$$
h_{12}(v,x) \partial_x \bar V_2(v,x,F) = f(v)\Big(h_{21}(v,x) - h_{22}(v,x)V_{\rho^*}(x,v)\Big).
$$
Using part (vi) of Assumption \ref{ass:main thm}, this implies that 
$$
h_{12}(v,x) =0, \qquad h_{21}(v,x)= h_{22}(v,x)V_{\rho^*}(x,v).
$$
Exploiting the surjectivity, this holds for all $(v,x)\in\interior(\A)$.
We evaluate $\partial_x\partial_v  \bar S(v,x,F) = \partial_v\partial_x  \bar S(v,x,F)$ again, but now for a general $(v,x)\in\interior(\A)$, and obtain
$$
\partial_x h_{11}(v,x) \big(F(v) - p\big) =  \partial_v h_{21}(v,x) \big(F(v) - p\big) + \partial_v h_{22}(v,x) \bar V_2(v,x,F).
$$
If $v=\VaR_p^-(F)$ and $x\neq \rho(F)$, we get that $\partial_v h_{22}(v,x)=0$. Exploiting the surjectivity, we get
$$
\partial_v h_{22}(v,x)=0, \qquad \partial_x h_{11}(v,x) =  \partial_v h_{21}(v,x), \qquad \text{ for all }(v,x)\in\interior(\A).
$$
Hence, we can write $h_{22}$ as a function of its second argument $x$ only. Now, we integrate the partial derivatives of the expected score.
\begin{align*}
\int \partial_v \bar S(v,x,F) \d v 
&= \int h_{11}(v,x) \big(F(v) - p\big)\d v 
= \int h_{11}(v,x) \int \big(\id_{\{y\le v\}} - p\big) \d F(y)\d v \\
&= \int G(y,x)\id_{\{y> v\}} +  \big(\id_{\{y\le v\}} - p\big)G(v,x) + a_1(y,x)\d F(y),
\end{align*}
where $\partial_v G(v,x)=h_{11}(v,x)$ and $a_1(y,x)$ is an integration constant not depending on $v$.
Similarly, 
\begin{align*}
&\int \partial_x \bar S(v,x,F) \d x 
= \int h_{22}(x)\Big(V_{\rho^*}(x,v)\big(F(v)-p\big) + \bar V_2(v,x,F)\Big) \d x \\
&= \int h_{22}(x) \int V_{\rho^*}(x,v)\big(\id_{\{y\le v\}} - p\big) + \id_{\{y> v\}} V_{\rho^*}(x,y)\d F(y) \d x \\
&= \int  \int h_{22}(x) V_{\rho^*}(x,v)\d x \big(\id_{\{y\le v\}} - p\big) + \id_{\{y> v\}}  \int h_{22}(x) V_{\rho^*}(x,y)\d x \d F(y) \\
&= \int  \Big(S_{\rho^*}(x,v) +a_2(v,y)\Big) \big(\id_{\{y\le v\}} - p\big) + \id_{\{y> v\}}  \Big( S_{\rho^*}(x,y) +a_3(v,y)\Big)   \d F(y),
\end{align*}
where $\partial_x S_{\rho^*}(x,y) = h_{22}(x) V_{\rho^*}(x,y)$ and $a_2(v,y)$ and $a_3(v,y)$ are integration constants not depending on $x$.
A comparison of the two results and an application of Proposition 1 in \cite{erratum} yields the form of $S$ in \eqref{eq:S general}. Since for any $F\in\M$, the function
$(v,y)\mapsto S(v,\rho(F),y)$ needs to be (strictly) $\M$-consistent for $\VaR_p^-$, $S_{\rho^*}(x,y)$ needs to be (strictly) increasing in $y$.
On the other hand, $(x,y)\mapsto S(\VaR_p^-(F),x,y)$ needs to be (strictly) $\M$-consistent for $\rho$, which is why $S_{\rho^*}$ needs to be (strictly) $\M$-consistent for $\rho^*$.
\end{proof}


\small 

  
 \begin{thebibliography}{10}

\bibitem[\protect\citeauthoryear{Acerbi and Sz\`ekely}{Acerbi and Sz\`ekely}{2014}]{AS14}
{Acerbi, C. and Sz\`ekely, B.} (2014). Back-testing expected shortfall. \emph{Risk Magazine}, \textbf{27}, 76--81.

\bibitem[\protect\citeauthoryear{Adrian and Brunnermeier}{Adrian and Brunnermeier}{2016}]{AdrianBrunnermeier2016}
{Adrian, T. and Brunnermeier, M. K.} (2016). CoVaR. \emph{American Economic Review}, \textbf{106}, 1705--1741.


\bibitem[\protect\citeauthoryear{Artzner et al.}{Artzner et al.}{1999}]{ADEH99}
{Artzner, P., Delbaen, F., Eber, J.-M. and Heath, D.} (1999). Coherent measures of risk. \emph{Mathematical Finance}, \textbf{9}(3), 203--228.


    \bibitem[\protect\citeauthoryear{{Basel Committee on Banking
  Supervision}}{{BCBS}}{2019}]{BASEL19}
{BCBS} (2019).
  {\em  Minimum Capital Requirements for Market Risk.  February 2019.}
 Basel Committee on Banking
  Supervision. Basel: Bank for International Settlements. \url{https://www.bis.org/bcbs/publ/d457.htm} 

\bibitem[\protect\citeauthoryear{Belles-Sampera et al.}{2014}]{BGS14}
Belles-Sampera, J., Guill\'en, M. and Santolino, M. (2014). Beyond value-at-risk: GlueVaR distortion risk measures. \emph{Risk Analysis}, \textbf{34}(1), 121--134.



\bibitem[\protect\citeauthoryear{Bellini and Bignozzi}{Bellini and Bignozzi}{2015}]{BB15} Bellini, F. and Bignozzi, V. (2015). On elicitable risk measures. \emph{Quantitative Finance}, \textbf{15}(5), 725--733.


\bibitem[\protect\citeauthoryear{Bellini et al.}{2014}]{BelliniETAL2014} Bellini, F., Klar, B., M\"uller, A. and Rossaza Gianin, E.  (2014). Generalized quantiles as risk measures. \emph{Insurance: Mathematics and Economics}, \textbf{54}, 41--48.

%\bibitem[\protect\citeauthoryear{Brehmer and Strokorb}{2019}]{BrehmerStrokorb2019}
%Brehmer, J. and Strokorb, K. (2019). Why scoring functions cannot assess tail properties. \emph{Electronic Journal of Statistics}.
%\textbf{13}, 4015--4034.

\bibitem[\protect\citeauthoryear{Cont at al.}{Cont
  et~al.}{2025}]{ContETAL2022}
Cont, R., Cucuringu, M., Xu, R., Zhang, C. (2025).
 Tail-GAN: Learning to simulate tail risk scenarios.
 {\em Management Science}, forthcoming. 

\bibitem[\protect\citeauthoryear{Cont, Deguest and He}{Cont
	et~al.}{2013}]{CDH13}
Cont, R., Deguest, R. and He, X. D. (2013).
Loss-based risk measures.
{\em Statistics and Risk Modeling}, {\bf 30}, 133--167.


\bibitem[\protect\citeauthoryear{Cont, Deguest and Scandolo}{Cont et~al.}{2010}]{CDS10}
Cont, R., Deguest, R. and Scandolo, G. (2010).
 Robustness and sensitivity analysis of risk measurement procedures.
 {\em Quantitative Finance}, {\bf 10}(6), 593--606.


 
 \bibitem[\protect\citeauthoryear{Dawid}{1998}]{Dawid1998} {Dawid, A.\ P.} (1998): \newblock \emph{Coherent Measures of Discrepancy, Uncertainty and Dependence, With Applications to Bayesian Predictive Experimental Design}.
\newblock Research Report 139, University College London, Dept. of Statistical Science.


\bibitem[\protect\citeauthoryear{Delbaen}{2012}]{D12} {Delbaen, F.} (2012). \newblock \emph{Monetary Utility Functions}.
\newblock Osaka University Press, Osaka.

\bibitem[\protect\citeauthoryear{Delbaen et al.}{2016}]{DBBZ16} Delbaen, F.,  Bellini, F.,  Bignozzi, V. and  Ziegel, J. (2016). Risk measures with convex level sets. \emph{Finance and Stochastics}, \textbf{20}(2), 433--453.

\bibitem[\protect\citeauthoryear{Dhaene et al.}{Dhaene et al.}{2002}]{DDGKV02}
{Dhaene, J., Denuit, M., Goovaerts, M. J., Kaas, R. and Vyncke, D.} (2002): {The concept of comonotonicity in actuarial science and finance: Theory}. {\em Insurance: Mathematics and Economics},
{\bf 31}(1), 3--33.

\bibitem[\protect\citeauthoryear{Dimitriadis et al.}{2024}]{DFZ23}
Dimitriadis, T., Fissler, T. and Ziegel, J. (2024). Osband's principle for identification functions. \emph{Statistical Papers}, \textbf{65},  1125--1132.

\bibitem[\protect\citeauthoryear{Dimitriadis et al.}{2024}]{DFZ24}
Dimitriadis, T., Fissler, T. and Ziegel, J. (2024). Characterizing $M$-estimators. \emph{Biometrika}, \textbf{111}(1), 339--346.

\bibitem[\protect\citeauthoryear{EIOPA}{EIOPA}{2011}]{E11}
EIOPA (2011). Equivalence assessment of the Swiss supervisory system in relation to articles 172, 227
and 260 of the Solvency II Directive, \emph{EIOPA-BoS-11-028}. %available at
%\url{https://eiopa.europa.eu/consultations/consultation-papers/index.html}.

\bibitem[\protect\citeauthoryear{Embrechts et al.}{2018}]{ELW18}
Embrechts, P., Liu, H. and Wang, R. (2018). Quantile-based risk sharing. \emph{Operations Research}, \textbf{66}(4), 936--949.



\bibitem[\protect\citeauthoryear{Embrechts et al.}{Embrechts et al.}{2021}]{EMWW21}
 {Embrechts, P., Mao, T., Wang, Q. and Wang, R.} (2021). Bayes risk, elicitability, and the Expected Shortfall. \emph{Mathematical Finance}, \textbf{31}, 1190--1217. 
 
\bibitem[{Emmer {et~al.}(2015)Emmer, Kratz and Tasche}]{EmmerKratzTasche2015}
Emmer S., Kratz M. and Tasche D. (2015). What is the best risk measure in practice? A
  comparison of standard measures. \emph{Journal of Risk}, \textbf{18}, 31--60.
 
 
\bibitem[\protect\citeauthoryear{Fissler et al.}{2021}]{FFHR2021}
Fissler, T. Frongillo, R., Hlavinov\'a, J. and Rudloff, B. (2021). Forecast evaluation of quantiles, prediction intervals, and other set-valued functionals. \emph{Electronic Journal of Statistics},
\textbf{15}, 1034--1084.



\bibitem[\protect\citeauthoryear{Fissler and Hoga}{2024}]{FisslerHoga2023}
Fissler, T. and Hoga, Y.  (2024). Backtesting systemic risk forecasts using multi-objective elicitability. \emph{Journal of Business \& Economic Statistics}, \textbf{42}(2), 485--498.

\bibitem[\protect\citeauthoryear{Fissler et al.}{2023}]{FMW23}
Fissler, T. Merz, M. and W\"uthrich, M. (2023). Deep quantile and deep composite model regression. \emph{Insurance: Mathematics and Economics}.
\textbf{109}, 94--112.



\bibitem[\protect\citeauthoryear{Fissler and Ziegel}{2016}]{FZ16}
Fissler, T. and Ziegel, J. F.  (2016). Higher order elicitability and Osband's principle. \emph{Annals of Statistics}, \textbf{44}(4), 1680--1707.

\bibitem[\protect\citeauthoryear{Fissler and Ziegel}{2019}]{FisslerZiegel2019}
Fissler, T. and Ziegel, J. F.  (2019). Order-sensitivity and equivariance of scoring functions. \emph{Electronic Journal of Statistics}, \textbf{13}(1), 1166--1211.

\bibitem[\protect\citeauthoryear{Fissler and Ziegel}{2021a}]{erratum}
Fissler, T. and Ziegel, J. F.  (2021a). Correction note: Higher order elicitability and Osband's principle. \emph{Annals of Statistics}, \textbf{49}(1), 614--614.

\bibitem[\protect\citeauthoryear{Fissler and Ziegel}{2021b}]{FZ21} 
Fissler, T. and Ziegel, J. F. (2021b). Elicitability of Range Value at Risk.  \emph{Statistics and Risk Modeling}, \textbf{38}(1-2), 25--46.


\bibitem[\protect\citeauthoryear{Fissler et al.}{2016}]{FZG16} 
Fissler, T. and Ziegel, J. F. and Gneiting, T. (2016) Expected Shortfall is jointly elicitable with Value at Risk, Implications for backtesting.
\emph{Risk Magazine}, January 2016, 58--61. 
 

\bibitem[\protect\citeauthoryear{F\"{o}llmer and Schied}{2016}]{FS16}
 {F\"{o}llmer, H. and Schied, A.} (2016).
 \emph{Stochastic Finance: An Introduction in Discrete Time.} {Walter de Gruyter, Berlin}, Fourth Edition.

 


\bibitem[\protect\citeauthoryear{Frongillo and Kash}{Frongillo and Kash}{2021}]{FK21}
Frongillo, R. and Kash, I. A. (2021). Elicitation complexity of statistical properties. \emph{Biometrika}, \textbf{108}(4), 857--879.
 

\bibitem[\protect\citeauthoryear{Furman and Landsman}{Furman and Landsman}{2006}]{FL06}
{Furman, E. and Landsman, Z.} (2006).
Tail variance premium with applications for elliptical portfolio of risks. 
\textit{ASTIN Bulletin}, \textbf{36}, 433--462.

\bibitem[\protect\citeauthoryear{Furman et al.}{Furman et al.}{2017}]{FWZ17}
{Furman, E., Wang, R. and Zitikis, R.} (2017). Gini-type measures of risk and variability: Gini shortfall, capital allocation and heavy-tailed risks. \emph{Journal of Banking and Finance}, \textbf{83}, 70--84.

\bibitem[\protect\citeauthoryear{Girardi and Tolga Erg\"un}{Girardi and Tolga Erg\"un}{2013}]{GT2013}
{Girardi, G. and Tolga Erg\"un, A.} (2013).
   Systemic risk measurement: Multivariate GARCH estimation of CoVaR.
  {\em Journal of Banking and Finance},~{\bf 37\/}, 3169--3180.
 
\bibitem[\protect\citeauthoryear{Gneiting}{Gneiting}{2011a}]{G11}
{Gneiting, T.} (2011a).
   Making and evaluating point forecasts.
  {\em Journal of the American Statistical Association\/},~{\bf 106\/}(494), 746--762.

 \bibitem[\protect\citeauthoryear{Gneiting}{Gneiting}{2011b}]{Gneiting2011b}
{Gneiting, T.} (2011b).
   Quantiles as optimal point forecasts.
  {\em International Journal of Forecasting\/},~{\bf 27\/}(2), 197--207.
  
   \bibitem[\protect\citeauthoryear{Gneiting and Raftery}{Gneiting and Raftery}{2007}]{GneitingRaftery2007}
{Gneiting, T. and Raftery, A. E.} (2007).
   Strictly proper scoring rules, prediction, and estimation.
  {\em Journal of the American Statistical Association\/},~{\bf 102\/}(477), 359--378.
  
   \bibitem[\protect\citeauthoryear{Gneiting and Ranjan}{Gneiting and Ranjan}{2011}]{GneitingRanjan2011}
{Gneiting, T. and Ranjan, R.} (2011).
   Comparing density forecasts using threshold-and quantile-weighted scoring rules.
  {\em Journal of Business \& Economic Statistics\/},~{\bf 29\/}(3), 411--422.
  
     \bibitem[\protect\citeauthoryear{Holzmann and Klar}{Holzmann and Klar}{2017}]{HolzmannKlar2017}
{Holzmann, H. and Klar, B.} (2017).
   Focusing on regions of interest in forecast evaluation.
  {\em The Annals of Applied Statistics\/},~{\bf 11\/}(4), 2404--2431.
  
  
   \bibitem[\protect\citeauthoryear{Huber and Ronchetti}{Huber and Ronchetti}{2009}]{HuberRonchetti2009}
{Huber, P. J. and Ronchetti, E. M.} (2009).
   \textit{Robust Statistics}, 2nd ed., John Wiley \& Sons, Hoboken.

\bibitem[\protect\citeauthoryear{Jarrow}{Jarrow}{2002}]{J02}
   Jarrow, R. (2002). Put option premiums and coherent risk measures. \emph{Mathematical Finance}, \textbf{12}(2), 135--142.
   
     \bibitem[\protect\citeauthoryear{Koenker}{Koenker}{2005}]{Koenker2005}
{Koenker, R.} (2005).
   \textit{Quantile Regression}, Cambridge University, Cambridge.
 

\bibitem[\protect\citeauthoryear{Kou and Peng}{2016}]{KP16}
{Kou, S. and Peng, X.} (2016). On the measurement of economic tail risk. \emph{Operations Research}, \textbf{64}(5), 1056--1072.

\bibitem[\protect\citeauthoryear{Lambert et al.}{2008}]{LPS08}
Lambert, N. S., Pennock, D. M. and Shoham, Y. (2008). Eliciting properties of probability distributions. In \emph{Proceedings of the 9th ACM Conference on Electronic Commerce}, 129-138. ACM.

 


\bibitem[\protect\citeauthoryear{Liu et al.}{2022}]{LMWW22}
Liu, F., Mao, T., Wang, R. and Wei, L. (2022). Inf-convolution, optimal allocations, and model uncertainty for tail risk measures. \emph{Mathematics of Operations Research}, , \textbf{47}(3), 2494--2519.


\bibitem[\protect\citeauthoryear{Liu and Wang}{2021}]{LW21}
Liu, F. and Wang, R. (2021). A theory for measures of tail risk.  \emph{Mathematics of Operations Research},  \textbf{46}(3), 
1109--1128.

 
 \bibitem[\protect\citeauthoryear{McNeil et al.}{McNeil et al.}{2015}]{QRM15}
{McNeil, A. J., Frey, R. and Embrechts, P.} (2015). \emph{Quantitative
Risk Management: Concepts, Techniques and Tools}. Revised Edition.  Princeton, NJ:
Princeton University Press.

 \bibitem[\protect\citeauthoryear{Nau}{1985}]{Nau1985}
Nau, R. F. (1985). Should scoring rules be ``effective"?  \emph{Management Science},  \textbf{31}, 527--535.





\bibitem[\protect\citeauthoryear{Newey and McFadden}{1994}]{NeweyMcFadden1994}
 {Newey, W. K. and McFadden, D.} (1994).
 Large sample estimation and hypothesis testing. In
 \emph{Handbook of
Econometrics,} 
vol. 4, R. F. Engle \& D. McFadden, eds. Amsterdam: Elsevier, pp. 2111--245.


 \bibitem[\protect\citeauthoryear{Newey and Powell}{1987}]{NeweyPowell1987}
Newey, W. K. and Powell, J. L. (1987). Asymmetric least squares estimation and testing. 
\emph{Econometrica},  \textbf{55}, 819--847.


 \bibitem[\protect\citeauthoryear{Nolde and Zhang}{2020}]{NoldeZhang2020}
Nolde, N., and Zhang, J. (2020). Conditional extremes in asymmetric financial markets.  \emph{Journal of Business \& Economic Statistics},  \textbf{38}, 201--213.

 \bibitem[\protect\citeauthoryear{Nolde and Ziegel}{2017}]{NoldeZiegel2017}
Nolde, N., and Ziegel, J. F. (2017). Elicitability and backtesting: Perspectives for banking regulation.  \emph{The Annals of Applied Statistics},  \textbf{11}(4), 1833--1874.

 \bibitem[\protect\citeauthoryear{Osband}{1985}]{Osband1985}
Osband, K. H. (1985). 
Providing incentives for better cost forecasting. \emph{PhD Thesis, University of California,
Berkeley}. \url{https://doi.org/10.5281/zenodo.4355667}.


\bibitem[\protect\citeauthoryear{Rockafellar and  Uryasev}{Rockafellar and  Uryasev}{2002}]{RU02}
Rockafellar, R. T. and Uryasev, S. (2002). Conditional value-at-risk for general loss distributions. \emph{Journal of Banking and Finance}, \textbf{26}(7), 1443--1471.

\bibitem[\protect\citeauthoryear{Steinwart et al.}{Steinwart et al.}{2014}]{SteinwartETAL2014}
 Steinwart, I., Pasin, C., Williamson, R. and Zhang, S. (2014). Elicitation and identification of properties. \emph{JMLR Workshop and Conference Proceedings}, \textbf{35}, 1--45.

\bibitem[\protect\citeauthoryear{Tsanakas}{2009}]{T09}
Tsanakas, A. (2009). To split or not to split: Capital allocation with convex risk measures. \emph{Insurance: Mathematics and Economics}, \textbf{44}(2), 268--277.



% \bibitem[\protect\citeauthoryear{Wang et al.}{Wang et al.}{2015}]{WBT15}
% {Wang, R., Bignozzi, V.  and  Tsakanas, A.} (2015).
%  How superadditive can a risk measure be?
%  \emph{SIAM Journal on Financial Mathematics}, \textbf{6}, 776--803.

 \bibitem[\protect\citeauthoryear{Wang et al.}{Wang et al.}{2025}]{WWZ25}
Wang, Q., Wang, R. and Ziegel, J.  (2025). E-backtesting. \emph{Management Science}, forthcoming.


\bibitem[\protect\citeauthoryear{Wang and Wei}{Wang and Wei}{2020}]{WW20} Wang, R. and Wei, Y. (2020). Risk functionals with convex level sets. \emph{Mathematical Finance},  \textbf{30}(4), 1337--1367.
 

% \bibitem[\protect\citeauthoryear{Wang et al.}{Wang et al.}{2020}]{WWW20}
% Wang, R., Wei, Y. and Willmot, G. E. (2020). Characterization, robustness and aggregation of signed Choquet integrals. \emph{Mathematics of Operations Research}, \textbf{45}(3), 993--1015.



\bibitem[\protect\citeauthoryear{Wang and Ziegel}{2015}]{WZ15}
 {Wang, R. and Ziegel, J.} (2015). Elicitable distortion risk measures: A concise proof.  \emph{Statistics and Probability Letters}, \textbf{100}, 172--175.

\bibitem[\protect\citeauthoryear{Weber}{Weber}{2006}]{W06}
{Weber, S.} (2006). Distribution-invariant risk measures, information, and dynamic consistency. \emph{Mathematical Finance}, \textbf{16}, 419--441.
 
\bibitem[\protect\citeauthoryear{Ziegel}{Ziegel}{2016}]{Z16}
Ziegel, J. (2016). Coherence and elicitability. {\em Mathematical Finance}, \textbf{26}, 901--918.

\end{thebibliography}
\end{document}